\documentclass{llncs}
\usepackage{epsfig,latexsym}

\usepackage{color}
\usepackage{algorithmic}
\usepackage{algorithm}
\usepackage{xfrac}
\usepackage{epsfig}
\usepackage{amssymb}
\usepackage{graphicx}
\usepackage{url}
\usepackage{amsmath}
\usepackage{hyperref}
\usepackage{slashbox}
\usepackage[font=footnotesize]{subcaption}
\newcommand{\newstuf}[1]{{\color{black} #1}} 
\newcommand{\commentout}[1]{}

\usepackage{wrapfig}

\begin{document}
\title{Metric tree-like structures in real-life networks: \\  an empirical study}
\author{Muad Abu-Ata  \and Feodor F. Dragan}
\institute{Algorithmic Research Laboratory, Department of Computer Science \\  Kent State University,
Kent, OH 44242, USA  \\ {\em \{mabuata,dragan\}@cs.kent.edu}}

\maketitle

\begin{abstract}
Based on solid theoretical foundations, we present strong evidences that a number of real-life networks, taken from different domains like Internet measurements, biological data, web graphs, social and collabora\-tion networks, exhibit tree-like structures from a metric point of view.
We investigate few graph parameters, namely, the tree-distortion and the tree-stretch, the tree-length and the tree-breadth, the Gro\-mov's hyperbolicity, the cluster-diameter and the cluster-radius in a layering partition of a graph, which capture and quantify this phenomenon of
being metrically close to a tree. \newstuf{By bringing all those parameters together, we not only provide efficient means for detecting such metric tree-like structures in large-scale networks but also show how such structures can be used, for example, to efficiently and compactly encode approximate distance and almost shortest path information and to fast and accurately estimate diameters and radii of those networks. Estimating the diameter and the radius of a graph or distances between its arbitrary vertices are fundamental primitives in many data and graph mining algorithms.}
\end{abstract}

\section{Introduction}\label{sec:intro}
Large networks are everywhere. Can we understand their structure
and exploit it? For example, understanding key structural properties of large-scale data networks is crucial for analyzing and
optimizing their performance, as well as improving their reliability and security \cite{narayan2011large}. In prior empirical  and theoretical studies researchers have mainly focused on features like small world phenomenon, power law degree distribution, navigability, high clustering coefficients, etc. (see
\cite{barabasi99emergence,barabasi2000scalefree,Boguna2009,DBLP:journals/im/ChungL03,DBLP:conf/sigcomm/FaloutsosFF99,
DBLP:conf/stoc/Kleinberg00,DBLP:conf/nips/Kleinberg01,DBLP:journals/im/LeskovecLDM09,Watts-Colective-1998}). Those nice features
were observed 
in many real-life complex networks and 
graphs arising in Internet applications, in biological and social sciences, in chemistry and physics. Although those features are interesting and important, as it is noted in \cite{narayan2011large}, the impact of intrinsic geometrical and topological features of large-scale data networks on performance, reliability and security is of much greater importance.

Recently, a few papers explored a little-studied before geometric characteristic of real-life networks, namely the
{\em hyperbolicity} (sometimes called also the {\em global curvature}) of the network (see, e.g., \cite{DBLP:conf/icdm/AdcockSM13,conf/isaac/ChenFHM12,conf/nca/MontgolfierSV11,Kennedy2013Arch,narayan2011large,DBLP:journals/ton/ShavittT08}). It was shown that a number of data networks, including Internet application networks, web networks, collaboration networks, social networks, and others, have small hyperbolicity.
It was suggested in \cite{narayan2011large} that property, observed in real-life networks,  that traffic between nodes tends to go through a
relatively small core of the network, as if the shortest path between them is curved inwards, may be due to global curvature
of the network. Furthermore, paper \cite{Kennedy2013Arch} proposes that ``hyperbolicity in conjunction with other local characteristics of networks, such as
the degree distribution and clustering coefficients, provide a more complete unifying picture of networks, and helps classify in a parsimonious way what is otherwise a bewildering and complex array of features and characteristics specific to each natural and man-made network".

The hyperbolicity of a graph/network can be viewed as a measure of how close a graph is to a tree metrically; the smaller the hyperbolicity of a graph is the closer it is metrically to a tree. Recent empirical results of \cite{DBLP:conf/icdm/AdcockSM13,conf/isaac/ChenFHM12,conf/nca/MontgolfierSV11,Kennedy2013Arch,narayan2011large,DBLP:journals/ton/ShavittT08} on hyperbolicity suggest that many
real-life complex networks and graphs may possess tree-like structures from a metric point of view.

In this paper, we substantiate this claim through analysis of a collection of real data
networks. We investigate few more, recently introduced graph parameters, namely, the {\em tree-distortion} and the {\em
tree-stretch} of a graph, the {\em tree-length} and the {\em
tree-breadth} of a graph, the Gromov's {\em hyperbolicity} of a
graph, the {\em cluster-diameter} and the {\em cluster-radius} in a
{\em layering partition} of a graph. All these parameters are trying to capture and quantify this phenomenon of being metrically close to a tree 
and can be used to measure metric tree-likeness of a real-life network. Recent advances in theory (see appropriate sections for details) allow us to calculate or accurately estimate those parameters for sufficiently large networks. By examining topologies of numerous publicly available networks, we demonstrate existence of metric tree-like structures in wide range of large-scale networks, from communication networks to various forms of social and biological
networks.

Throughout this paper we discuss these parameters and recently established relationships between them for unweighted and undirected graphs. It
turns out that all these parameters are at most constant or logarithmic factors apart from each other. Hence, a constant bound on one of them translates in a constant or almost constant bound on another. We say that a graph {\em has a tree-like structure from a metric point of view} (equivalently, {\em is metrically tree-like}) if anyone of those parameters is a small constant.

\newstuf{
Recently, paper \cite{DBLP:conf/icdm/AdcockSM13} pointed out that 
"although large informatics graphs such as social and information networks are often thought of as having hierarchical or tree-like structure, this assumption is rarely tested, and it has proven difficult to exploit this idea in practice; ... it is not clear whether such structure can be exploited for improved graph mining and machine learning ...".

In this paper, by bringing all those parameters together, we not only provide efficient means for detecting such metric tree-like structures in large-scale networks  but also show how such structures can be used, for example, to efficiently and compactly encode approximate distance and almost shortest path information and to fast and accurately estimate diameters and radii of those networks. Estimating accurately and quickly distances between arbitrary vertices of a graph is a fundamental primitive in many data and graph mining algorithms. 
}

Graphs that are metrically tree-like have many algorithmic advantages. They allow efficient approximate solutions for a number of optimization problems. For example, they admit a PTAS for the Traveling Salesman Problem~\cite{KrLe}, have an efficient approximate solution for the problem of covering and packing by balls~\cite{DBLP:conf/approx/ChepoiE07}, admit additive sparse spanners~\cite{ChDrEsRout,DoDrGaYa} and collective additive tree-spanners~\cite{DBLP:conf/sofsem/DraganA13}, enjoy efficient and compact approximate distance~\cite{ChDrEsRout,GaLy}  and routing~\cite{ChDrEsRout,DBLP:journals/jgaa/Dourisboure05} labeling schemes, have efficient algorithms for fast and accurate estimations of diameters
and radii~\cite{DBLP:conf/compgeom/ChepoiDEHV08}, etc.. We elaborate more on these results in appropriate sections.

\newstuf{
For the first time such metric parameters, as tree-length and tree-breadth, tree-distortion and tree-stretch, cluster-diameter  and  cluster-radius, were examined, and algorithmic advantages of having those parameters bounded by small constants were discussed for such a wide range of large-scale networks.

}

This paper is structured as follows. In Section \ref{sec:notions}, we give notations and basic notions used in the paper. In Section \ref{sec:datasets}, we describe our graph datasets. The next four sections are devoted to analysis of corresponding parameters measuring metric tree-likeness of our graph datasets: layering partition and its cluster-diameter and cluster-radius in Section \ref{sec:layer-partit};  hyperbolicity in Section \ref{sec:hyperbol}; tree-distortion in Section \ref{sec:td}; tree-breadth, tree-length and tree-stretch in Section \ref{sec:tb}. In each section we first give theoretical background on the parameter(s) and then present our experimental results. Additionally, an overview of implications of those results is provided.  In Section \ref{appl}, we further discuss algorithmic advantages for a graph to be metrically tree-like. Finally, in Section  \ref{sec:concl}, we give some concluding remarks.

\section{Notations and Basic Notions}\label{sec:notions}
All graphs in this paper are connected, finite, unweighted, undirected, loopless and without multiple edges. 
For a graph $G=(V,E)$, we use $n$ and $|V|$ interchangeably to denote the number of vertices in $G$. Also, we use $m$ and $|E|$ to denote the number of edges. The {\em length of a path} from a vertex $v$ to a vertex $u$ is the number of edges in the path. The {\em distance} $d_G(u,v)$ between vertices $u$ and $v$ is the length of the shortest path connecting $u$ and $v$ in $G$. The {\em ball} $B_r(s,G)$ of a graph $G$ centered at vertex $s \in V$ and with radius $r$ is the set of all vertices with distance no more than $r$ from $s$ (i.e., $B_r(s,G)=\{v\in V: d_G(v,s) \leq r \}$). We omit the graph name $G$ as in $B_r(s)$ if the context is about only one graph.

The \emph{diameter} $diam(G)$ of a graph $G=(V,E)$ is the largest distance between a pair of vertices in $G$, i.e., $diam(G)=\max_{u,v \in V}d_G(u,v)$. The \emph{eccentricity} of a vertex $v$, denoted by $ecc(v)$, is the largest distance from that vertex $v$ to any other vertex, i.e., $ecc(v)=\max_{u \in V} d_G(v,u)$.  The \emph{radius} $rad(G)$ of a graph $G=(V,E)$ is the minimum eccentricity of a vertex in $G$, i.e., $rad(G)=\min_{v \in V} \max_{u\in V}d_G(v,u)$. The \emph{center} $C(G)=\{c \in V: ecc(c)=rad(G)\}$ of a graph $G=(V,E)$ is the set of vertices with minimum eccentricity.

Definitions of graph parameters measuring metric tree-likeness of a graph, as well as notions and notations local to a section, are given in appropriate sections.

\section{Datasets}\label{sec:datasets}

Our datasets come from different domains like Internet measurements, biological datasets, web graphs, social and collaboration networks. Table~\ref{tab:datasets} shows basic statistics of our graph datasets. Each graph represents the largest connected component of the original graph as some datasets consist of one large connected component and many very small ones.

\begin{table}
\footnotesize
\begin{center}
\begin{tabular}{ | c | c | c | c | c |}
    \hline
  Graph & n= & m=  & diameter   & radius   \\
  $G=(V,E)$ &  $|V|$ &  $|E|$  &  $diam(G)$ &  $rad(G)$ \\ \hline \hline
  PPI~\cite{ppi} & 1458 & 1948 & 19 & 11\\ \hline
  Yeast~\cite{yeast} & 2224 & 6609 & 11 & 6 \\ \hline
  DutchElite~\cite{dutchElite} & 3621 & 4311 & 22  & 12 \\ \hline
  EPA~\cite{epa} &4253 &8953 &10 & 6 \\ \hline
  EVA~\cite{eva} &4475 &4664 &18 & 10 \\ \hline
  California~\cite{california} & 5925 & 15770 & 13  & 7 \\ \hline
  Erd\"os~\cite{erdos} & 6927 & 11850 & 4  & 2 \\ \hline
  Routeview~\cite{routeview} & 10515 & 21455 & 10 & 5 \\ \hline
  Homo  release 3.2.99~\cite{homo} & 16711& 115406 & 10 & 5 \\ \hline
  AS\_Caida\_20071105~\cite{AS-Caida-20071105} & 26475 & 53381 & 17 & 9 \\ \hline
  Dimes 3/2010~\cite{dimes} & 26424 & 90267 & 8  & 4 \\ \hline
  Aqualab 12/2007- 09/2008~\cite{aqualab} & 31845 & 143383 & 9  & 5 \\ \hline
  AS\_Caida\_20120601~\cite{AS-Caida-20120601} & 41203 & 121309 & 10 & 5 \\ \hline
  itdk0304~\cite{itdk0304} & 190914 & 607610 & 26  & 14 \\ \hline
  DBLB-coauth~\cite{dblbAmazon} & 317080 & 1049866 & 23 & 12 \\ \hline
  Amazon~\cite{dblbAmazon} & 334863 & 925872 & 47  & 24 \\
  \hline
\end{tabular}
\end{center}
\caption{Graph datasets and their parameters: number of vertices, number of edges, diameter, radius.}
\label{tab:datasets} \vspace*{-6mm}
\end{table}

\medskip

\noindent
\textbf{Biological Networks}

\noindent
\underline{PPI}~\cite{ppi}: It is a protein-protein interaction network in the yeast Saccharomyces Cerevisiae. Each node represents a protein with an edge representing an interaction between two proteins. Self loops have been removed from the original dataset. The dataset has been analyzed and described in~\cite{ppi}.

\noindent\underline{Yeast}~\cite{yeast}: It is a protein-protein interaction network in budding yeast. Each node represents a protein with an edge representing an interaction between two proteins. Self loops have been removed from the original dataset. The dataset has been analyzed and described in~\cite{yeast}.

\noindent\underline{Homo}~\cite{homo}: It is a dataset of protein and genetic interactions in Homo Sapiens (Human). Each node represents a protein or a gene. An edge represents an interaction between two proteins/genes. Parallel edges, representing different resources for an interaction, have been removed. The dataset is obtained from BioGRID, a freely accessible database/repositiory of physical and genetic interactions available at \url{http://www.thebiogrid.org}. The dataset has been analyzed and described in~\cite{homo}.

\medskip

\noindent
\textbf{Social and Collaboration Networks}

\noindent\underline{DutchElite}~\cite{dutchElite}: This is data on the administrative elite in Netherland, April 2006. Data collected and analyzed by De Volkskrant and Wouter de Nooy. A 2-mode network data representing person's membership in the administrative and organization bodies in Netherland in 2006. A node represents either a person or an organization body. An edge exists between two nodes if the person node belongs to the organization node.

\noindent\underline{EVA}~\cite{eva}: It is a network of interconnection between corporations where an edge exists between two companies (vertices) if one of them is the owner of the other company.

\noindent\underline{Erd\"os}~\cite{erdos}: It is a collaboration network with mathematician Paul Erd\"os. Each vertex represents an author with an edge representing a paper co-authorship between two authors.

\noindent\underline{DBLB-coauth}~\cite{dblbAmazon}: It is a co-authorship network of the DBLP computer science bibliography. Vertices of the network represent authors with edges connecting two authors if they published at least one paper together.
\medskip

\noindent
\textbf{Web Graphs}

\noindent\underline{EPA}~\cite{epa}: It is a dataset representing pages linking to \url{www.epa.gov} obtained from Jon Kleinberg's web page, \url{http://www.cs.cornell.edu/courses/cs685/2002fa/}. The pages were constructed by expanding a 200-page response set to a search engine query, as in the hub/authority algorithm. This data was collected some time back, so a number of the links may not exist anymore. The vertices of this graph dataset  represent web pages with edges representing links. The graph was originally directed. We ignored direction of edges to get undirected graph version of the dataset.

\noindent\underline{California}~\cite{california}: This graph dataset was also constructed by expanding a 200-page response set to a search engine query 'California', as in the hub/authority algorithm. The dataset was obtained from Jon Kleinberg's page, \url{http://www.cs.cornell.edu/courses/cs685/2002fa/}. The vertices of this graph dataset represent web pages with edges representing links between them. The graph was originally directed. We ignored direction of edges to obtain undirected graph version of the dataset.
\medskip

\noindent
\textbf{Internet Measurements Networks}

\noindent\underline{Routeview}~\cite{routeview}: It is an Autonomous System (AS) graph obtained by University of Oregon Route-views project using looking glass data and routing registry. A vertex in the dataset represents an AS with an edge linking two vertices if there is at least one physical link between them.

\noindent\underline{AS\_Caida}~\cite{AS-Caida-20071105,AS-Caida-20120601}: These are datasets of the Internet Autonomous Systems (AS) relationships derived from BGP table snapshots taken at 24-hour intervals over a 5-day period by CAIDA. The AS relationships available are customer-provider (and provider-customer, in the opposite direction), peer-to-peer, and sibling-to-sibling.

\noindent\underline{Dimes 3/2010}~\cite{dimes}: It is an AS relationship graph of the Internet obtained from Dimes. The Dimes project performs traceroutes and pings from volunteer agents (of about 1000 agent computers) to infer AS relation\-ships. A weekly AS snapshot is available.  The dataset \emph{Dimes 3/2010} represents a snapshot aggregated over the month of March, 2010. It provides the set of AS level nodes and edges that were found in that month and were seen at least twice.

\noindent\underline{Aqualab}~\cite{aqualab}: Peer-to-peer clients are used to collect traceroute paths which are used to infer AS intercon\-nections. Probes were made between December 2007 and September 2008 from approximately 992,000 P2P users in 3,700 ASes.


\noindent\underline{Itdk}~\cite{itdk0304}:  This is a dataset of Internet router-level graph where each vertex represents a router with an edge between two vertices if there is a link between the corresponding routers. The dataset snapshot is computed from ITDK0304 skitter and iffinder measurements. The dataset is provided by CAIDA for April 2003 (see \url{http://www.caida.org/data/active/internet-topology-data-kit}).

\medskip

\noindent
\textbf{Information network}

\noindent\underline{Amazon}~\cite{dblbAmazon}: It is an Amazon product co-purchasing network.
The vertices of the network represent products purchased from the Amazon website and the edges link ``commonly/frequently'' co-purchased products.

\section{Layering Partition, its Cluster-Diameter and Cluster-Radius} \label{sec:layer-partit}


Layering partition is a graph decomposition procedure that has been introduced in~\cite{DBLP:journals/jal/BrandstadtCD99,DBLP:journals/ejc/ChepoiD00} and has been used in~\cite{DBLP:journals/jal/BrandstadtCD99,DBLP:journals/ejc/ChepoiD00,ChepoiDNRV12} and~\cite{BaInSi} for embedding graph metrics into trees. It provides a central tool in our investigation.

A \emph{layering} of a graph $G=(V, E)$ with respect to a start vertex $s$ is the decomposition of $V$ into the layers (spheres) $L^i=\{u\in V:d_G(s,u)=i\},i=0,1,\dots,r$. A \emph{layering partition} $\mathcal{LP}(G,s)=\{L^i_1,\cdots,L^i_{p_i}:i=0,1,\dots,r\}$ of $G$ is a partition of each layer $L^i$ into clusters $L^i_1,\dots,L^i_{p_i}$ such that two vertices $u,v \in L^i$ belong to the same cluster $L^i_j$ if and only if they can be connected by a path outside the ball $B_{i-1}(s)$ of radius $i-1$ centered at $s$. See Fig. \ref{fig:layering-partition} for an illustration. A layering partition of a graph can be constructed in $O(n+m)$ time (see~\cite{DBLP:journals/ejc/ChepoiD00}).

\begin{figure}
\vspace*{-1.4cm}
        \centering
        \begin{subfigure}[b]{0.450\textwidth}
                \includegraphics[width=\textwidth]{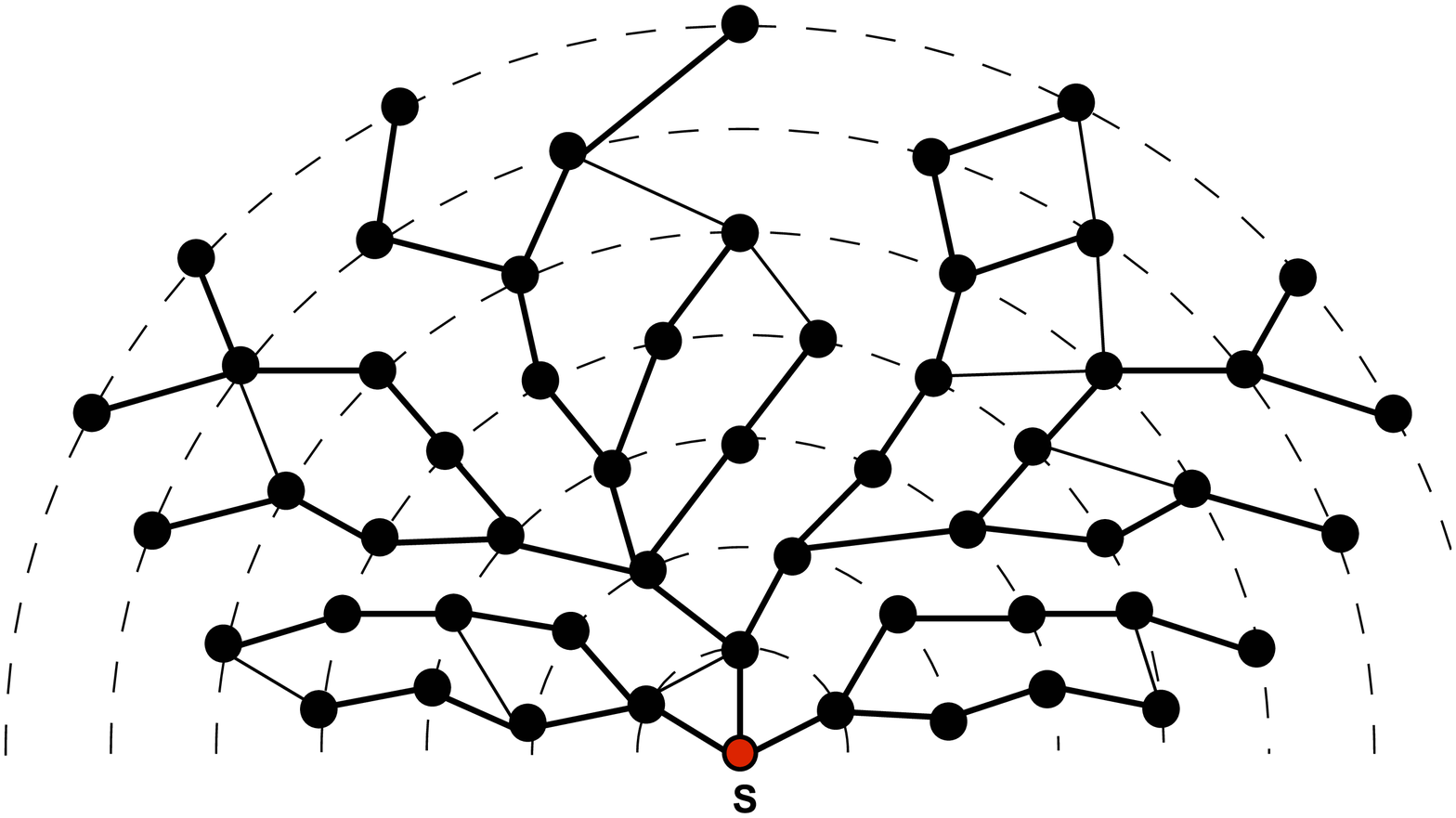}
                \vspace*{-.9cm}
                \caption{Layering of graph $G$ with respect to $s$.}
                \label{fig:layering}
        \end{subfigure}%
        \begin{subfigure}[b]{0.450\textwidth}
                \includegraphics[width=\textwidth]{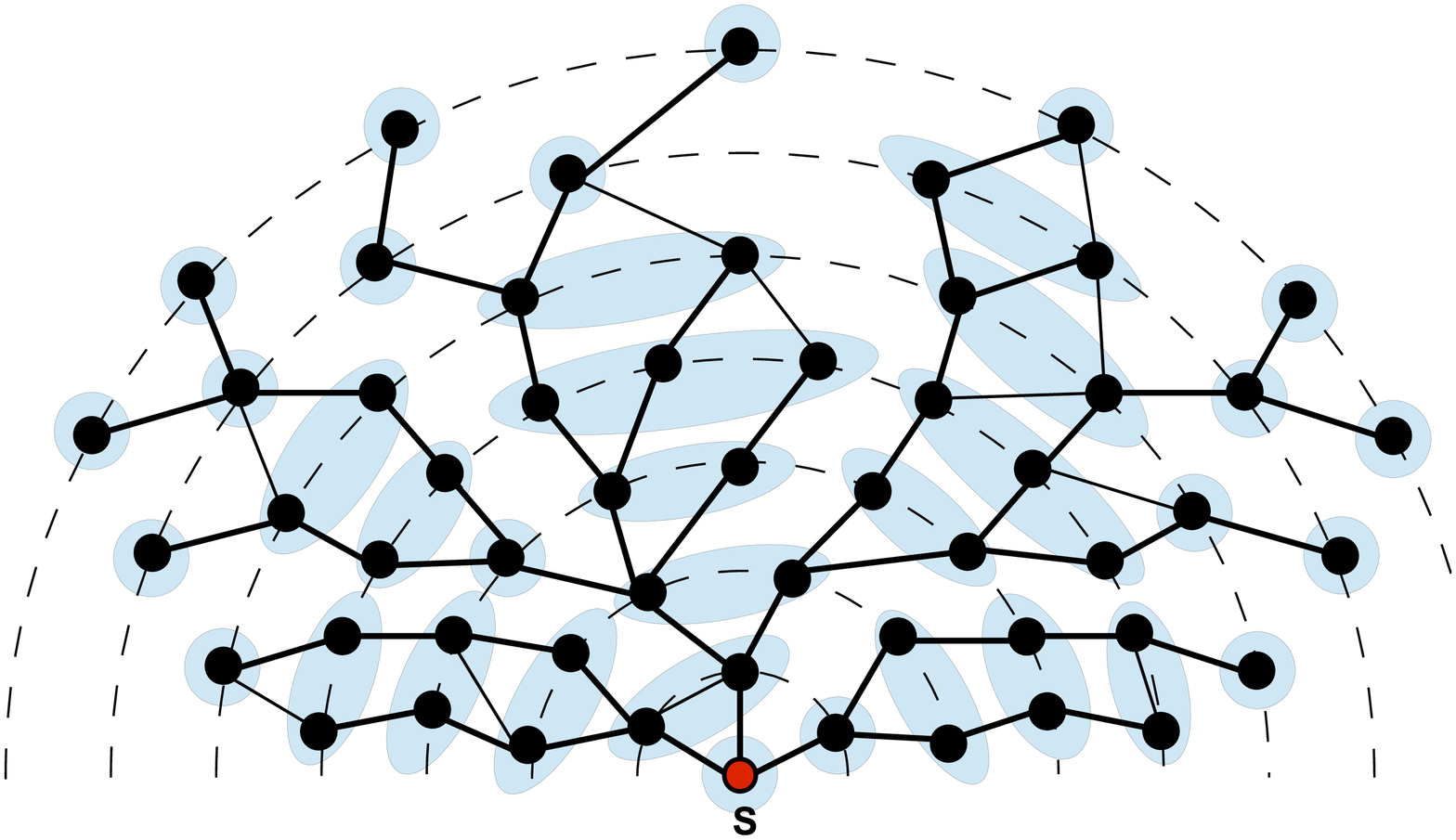}
                \vspace*{-.9cm}
                \caption{Clusters of the layering partition $\mathcal{LP}(G,s)$.}
                \label{fig:Layering-clusters}
        \end{subfigure}

         \begin{subfigure}[b]{0.450\textwidth} \vspace*{-.5cm}
                \includegraphics[width=\textwidth]{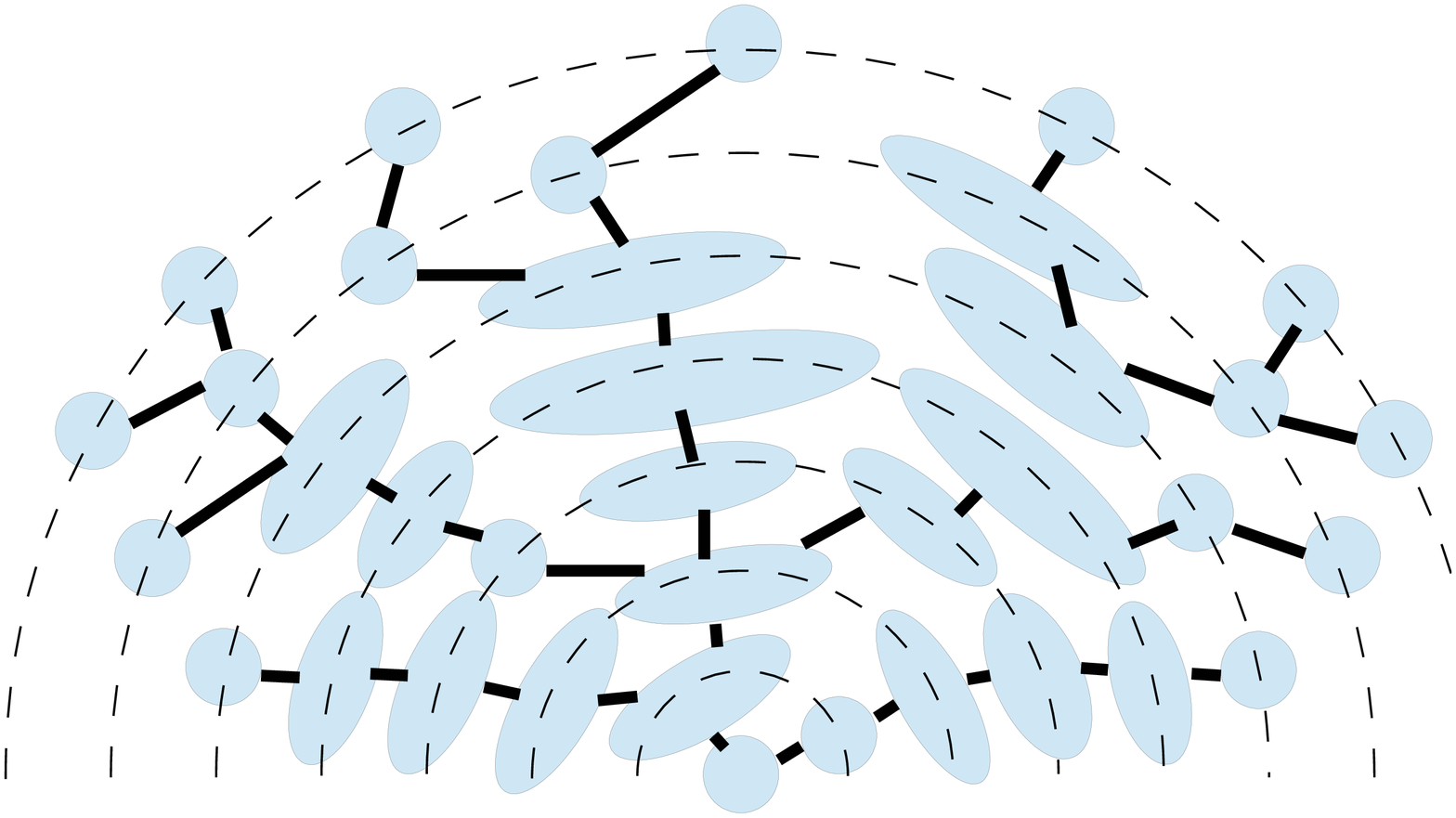}
                \vspace*{-.9cm}
                \caption{Layering tree $\Gamma(G,s)$.}
                \label{fig:gamma}
        \end{subfigure}
        \begin{subfigure}[b]{0.450\textwidth}\vspace*{-.5cm}
                \includegraphics[width=\textwidth]{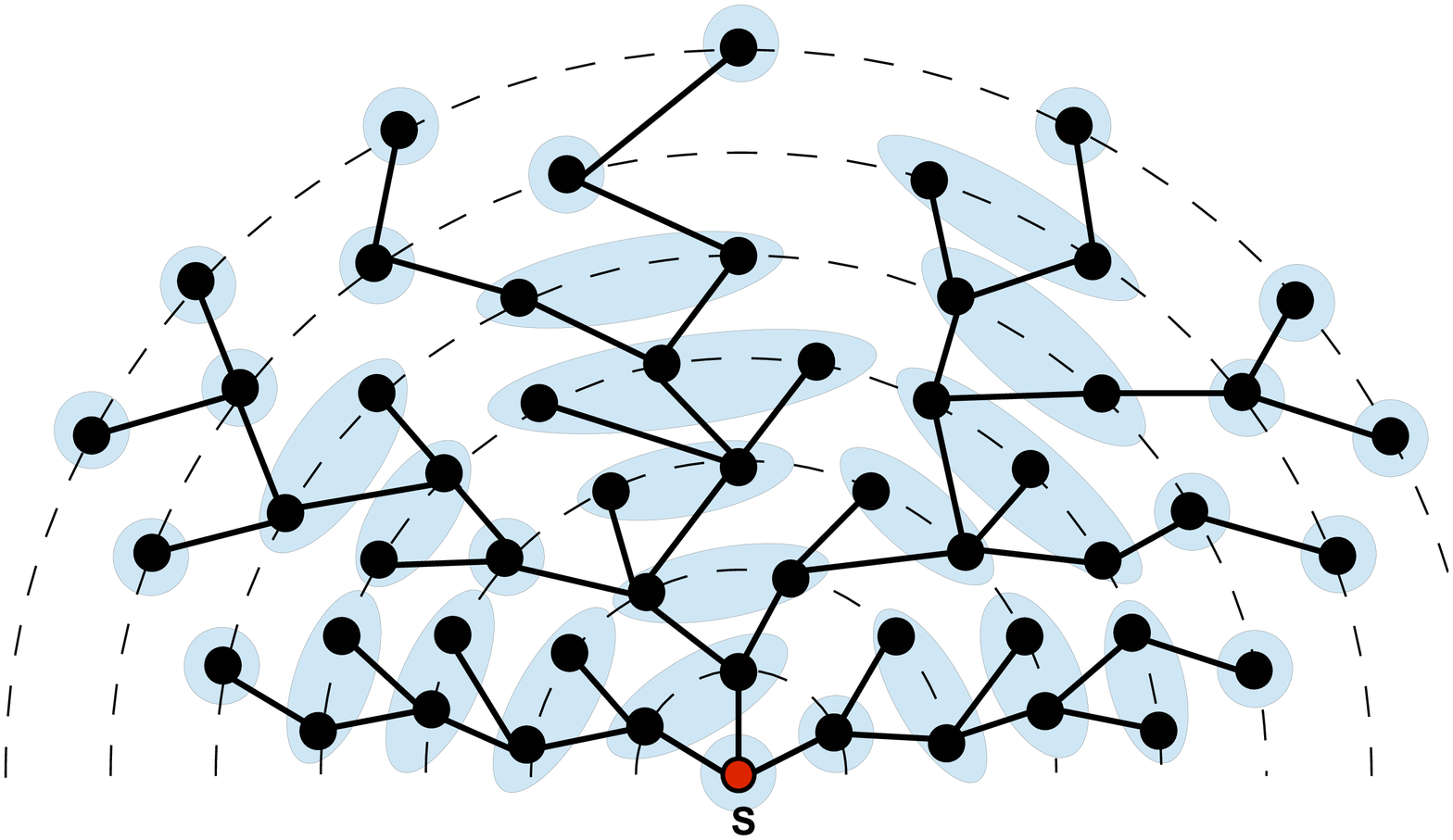}
                \vspace*{-.9cm}
                \caption{Canonic tree $H$ obtained from the layering partition.} 
                \label{fig:treeH}
        \end{subfigure}
        \caption{Layering partition and associated constructs.}\label{fig:layering-partition}
\end{figure}

A \emph{layering tree} $\Gamma(G,s)$ of a graph $G$ with respect to a layering partition $\mathcal{LP}(G,s)$  is the graph whose nodes are the clusters of $\mathcal{LP}(G,s)$ and two nodes $C=L_j^i$ and $C'=L_{j'}^{i'}$ are adjacent in $\Gamma(G,s)$ if and only if there exist a vertex $u \in C$ and a vertex $v\in C'$ such that $uv \in E$. It was shown in~\cite{DBLP:journals/jal/BrandstadtCD99} that the graph $\Gamma(G,s)$ is always a tree and, given a start vertex $s$,  can be constructed in $O(n+m)$ time~\cite{DBLP:journals/ejc/ChepoiD00}. Note that, for a fixed start vertex $s\in V$, the layering partition $\mathcal{LP}(G,s)$ of $G$ and its tree $\Gamma(G,s)$ are unique.

The \emph{cluster-diameter $\Delta_s(G)$ of layering partition $\mathcal{LP}(G,s)$ with respect to vertex $s$} is the largest diameter of a cluster in $\mathcal{LP}(G,s)$, i.e., $\Delta_s(G)=\max_{C \in \mathcal{LP}(G,s)} \max_{u,v\in C}d_G(u,v)$. The \emph{cluster-diameter $\Delta(G)$ of a graph $G$} is the minimum cluster-diameter over all layering partitions of $G$, i.e. $\Delta(G)=\min_{s \in V}\Delta_s(G)$.

The \emph{cluster-radius $R_s(G)$ of layering partition $\mathcal{LP}(G,s)$ with respect to a vertex $s$} is the smallest number $r$ such that for any cluster $C \in \mathcal{LP}(G,s)$ there is a vertex $v \in V$ with $C \subseteq B_r(v)$. The \emph{cluster-radius $R(G)$ of a graph $G$} is the minimum cluster-radius over all layering partitions of $G$, i.e., $R(G)=\min_{s \in V}R_s(G)$.

Clearly, in view of tree $\Gamma(G,s)$ of $G$, the smaller parameters $\Delta_s(G)$ and $R_s(G)$ of $G$ are, the closer graph $G$ is to a tree metrically.

Finding cluster-diameter $\Delta_s(G)$ and cluster-radius $R_s(G)$ for a given layering partition $\mathcal{LP}(G,s)$ of a graph $G$ requires $O(n m)$ time\footnote{The parameters $\Delta(G)$ and $R(G)$ can also be computed in total $O(n m)$ time for any graph $G$.}, although the construction of layering partition $\mathcal{LP}(G,s)$ itself, for a given vertex $s$, takes only $O(n+m)$ time. Since the diameter of any set is at least its radius and at most twice its radius, we have the following inequality: $$R_s(G) \leq \Delta_s(G) \leq 2R_s(G).$$


\begin{table}\vspace*{-6mm}
\footnotesize
\begin{center}
\begin{tabular}{ | c | c | c | c| c | p{2.4cm} | p{2.6cm} |}  
    \hline
  Graph     & n=    & diameter  & \# of clusters         & cluster-         & average diameter    & $\%$ of clusters \\
  $G=(V,E)$ & $|V|$ & $diam(G)$ & in $\mathcal{LP}(G,s)$ & diameter         & of clusters in      & having diameter 0  \\
            &       &           &                        & $\Delta_s(G)$    & $\mathcal{LP}(G,s)$ & or 1 (i.e., cliques)  \\ \hline \hline
  PPI & 1458  & 19 & 1017 &8 & ~~0.118977384& ~~97.05014749\% \\ \hline
  Yeast & 2224 & 11 &1838 &6 & ~~0.119575699 & ~~96.33558341\% \\ \hline
  DutchElite  & 3621  & 22 &2934 & 10& ~~0.070211316 &~~98.02317655\% \\ \hline
  EPA  &4253  &10 &2523 &6 & ~~0.06698375 & ~~98.5731272\%\\ \hline
  EVA  &4475 &18 &4266 &9 & ~~0.031879981 & ~~99.2030005\%\\ \hline
  California & 5925  & 13 &2939 &8 & ~~0.092208234 &~~97.141885\%\\ \hline
  Erd\"os  & 6927  & 4 &6288 &4 & ~~0.001113232 &~~99.9681934\% \\ \hline
  Routeview  & 10515  & 10 &6702 & 6 & ~~0.063264697 & ~~98.4482244\%\\ \hline
  Homo  release 3.2.99  & 16711 & 10 & 6817 &5 & ~~0.03432595 & ~~99.2518703\% \\ \hline
  AS\_Caida\_20071105  & 26475  & 17&17067 & 6 & ~~0.056424679 &~~98.5527626\%\\ \hline
  Dimes 3/2010  & 26424  & 8 &16065 &4 & ~~0.056582633 & ~~98.5434174\%\\ \hline
  Aqualab 12/2007- 09/2008  & 31845  & 9 &16287 & 6 & ~~0.05826733& ~~98.5816909\% \\ \hline
  AS\_Caida\_20120601  & 41203  & 10& 26562& 6 & ~~0.055568105 & ~~98.5731496\% \\ \hline
  itdk0304 & 190914  & 26 &89856 & 11 & ~~0.270377048 & ~~91.3851051\%  \\ \hline
  DBLB-coauth & 317080  & 23 &99828 &11 & ~~0.45350002 & ~~92.97091\% \\ \hline
  Amazon & 334863  & 47 &72278 & 21 & ~~0.489056144 & ~~86.049697\% \\
  \hline
\end{tabular}
\end{center}
\caption{Layering partitions of the datasets and their parameters. $\Delta_s(G)$ is the largest diameter of a cluster in $\mathcal{LP}(G,s)$, where $s$ is a randomly selected start vertex. For all datasets, the average diameter of a cluster is between 0 and 1. For most datasets,  more than 95\% of clusters are cliques.}
\label{tab:layering}
\vspace*{-0.5cm}\end{table}

In Table~\ref{tab:layering}, we show empirical results on layering partitions obtained for datasets described in Section \ref{sec:datasets}. For each graph dataset $G=(V,E)$, we randomly selected a start vertex $s$ and built layering partition $\mathcal{LP}(G,s)$ of $G$ with respect to $s$. For each dataset, Table~\ref{tab:layering} shows the cluster-diameter
$\Delta_s(G)$, the number of clusters in layering partition $\mathcal{LP}(G,s)$ and the average diameter of clusters in $\mathcal{LP}(G,s)$. It turns out that all graph datasets have small average diameter of clusters. Most clusters have diameter $0$ or $1$, i.e., they are essentially cliques (=complete subgraphs) of $G$. For most datasets,  more than 95\% of clusters are cliques.

To have a better picture on the overall distribution of diameters of clusters, in Table~\ref{tab:lay-diam-freq}, we show the frequencies of diameters of clusters for three sample datasets: PPI, Yeast, and AS\_Caida\_20071105. It is interesting to note that, in all 
datasets, the clusters with large diameters induce a connected subtree in the
tree $\Gamma(G,s)$. For example, in PPI, the cluster with diameter 8 is adjacent in $\Gamma(G,s)$ to all clusters with diameters 6 and 5. This may indicate that all those clusters are part of the well connected network core.

\begin{table}
    \footnotesize
    \centering
   \begin{subtable}{.3\textwidth}
        \centering
        \begin{tabular}{ | c | c |c|}
            \hline
        diameter   & ~frequency~ & relative \\
        ~of a cluster~ & & ~frequency~\\ \hline\hline
            0 & 966 &0.9499\\ \hline
            1 & 21 & 0.0206\\ \hline
            2 & 14 & 0.0138 \\ \hline
            3 & 5  & 0.0049 \\ \hline
            4 & 5 & 0.0049\\ \hline
            5 & 1 & 0.0001\\ \hline
            6 & 4 &0.0039\\ \hline
            7 & 0 &0\\ \hline
            8 & 1 &0.0001\\
            \hline
        \end{tabular}
    \caption{PPI}
    \end{subtable}
    \hspace*{2mm}
   \begin{subtable}{.3\textwidth}
   \centering
    \begin{tabular}{ | c | c |c|}
        \hline
        diameter   & ~frequency~ & relative \\
        ~of a cluster~ & & ~frequency~\\ \hline\hline
        0 & 981 &0.946\\ \hline
        1 & 18 &0.0174\\ \hline
        2 & 23 &0.0223 \\ \hline
        3 & 6 &0.0058 \\ \hline
        4 & 5 &0.0048\\ \hline
        5 & 2 &0.0019\\ \hline
        6 & 2 &0.0019 \\
        \hline
    \end{tabular}
   \caption{Yeast}
   \end{subtable}
   \hspace*{2mm}
    \begin{subtable}{.3\textwidth}
    \centering
    \begin{tabular}{ | c | c |c|}
        \hline
        diameter   & ~frequency~ & relative \\
        ~of a cluster~ & & ~frequency~\\ \hline\hline
        0 & 16459 &0.9644 \\ \hline
        1 & 361 & 0.0216\\ \hline
        2 & 174 &0.0102\\ \hline
        3 & 46 &0.0027\\ \hline
        4 & 21 &0.0012\\ \hline
        5 & 4 &0.0002 \\ \hline
        6 & 2 & 0.0001\\
        \hline
    \end{tabular}
   \caption{AS\_Caida\_20071105}
   \end{subtable}
   \caption{Frequency of diameters of clusters in layering partition $\mathcal{LP}(G,s)$ (three datasets).}
   \label{tab:lay-diam-freq}
\vspace*{-0.5cm}\end{table}

Most of the graph parameters discussed in this paper could be related to a special tree $H$ introduced in~\cite{ChepoiDNRV12} and produced from a layering partition of a graph $G$.

\textbf{Canonic tree} $\mathbf{H}$: A tree $H=(V,F)$ of a graph $G=(V,E)$, called a {\em canonic tree of $G$}, is constructed from a layering partition $\mathcal{LP}(G,s)$ of $G$ by identifying for each cluster $C=L^i_j \in \mathcal{LP}(G,s)$ an arbitrary vertex $x_C \in L_{i-1}$  which has a neighbor in $C = L^i_j$ and by making $x_C$ adjacent in $H$ with all vertices $v\in C$ (see Fig. \ref{fig:treeH} for an illustration). Vertex $x_C$ is called the support vertex for cluster $C= L^i_j$. It was shown in~\cite{ChepoiDNRV12} that tree $H$ for a graph $G$ can be constructed in $O(n+m)$ total time.

The following statement from~\cite{ChepoiDNRV12} relates the cluster-diameter of a layering partition of $G$ with embedability of graph $G$ into the tree $H$.
\begin{proposition} [\cite{ChepoiDNRV12}]
\label{lem:cluster-diam}
For every graph $G=(V,E)$ and any vertex $s$ of $G$, $$\forall x,y \in V, ~~d_H(x,y)-2 \leq d_G(x,y) \leq d_H(x,y)+\Delta_s(G).$$
\end{proposition}

The above proposition shows that the distortion of embedding of a graph $G$ into tree $H$
is additively bounded by $\Delta_s(G)$, the largest diameter of a cluster in a layering partition of $G$.
This result confirms that the smaller cluster-diameter $\Delta_s(G)$ (cluster-radius $R_s(G)$) of $G$ is, the closer graph $G$ is to a tree metric. Note that trees have cluster-diameter and cluster-radius equal to $0$. Results similar to Proposition \ref{lem:cluster-diam} were used in~\cite{DBLP:journals/jal/BrandstadtCD99} to embed a chordal graph to a tree with an additive distortion at most 2, in~\cite{DBLP:journals/ejc/ChepoiD00} to embed a $k$-chordal graph to a tree with an additive distortion at most $k/2 +2$, and in~\cite{ChepoiDNRV12} to obtain a 6-approximation algorithm for the problem of optimal non-contractive embedding of an unweighted graph metric into a weighted tree metric. For every {\em chordal graph} $G$ (a graph whose largest induced cycles have length 3),  $\Delta_s(G) \leq 3$ and $R_s(G)\leq 2$ hold~\cite{DBLP:journals/jal/BrandstadtCD99}. For every {\em $k$-chordal graph} $G$ (a graph whose largest induced cycles have length $k$), $\Delta_s(G) \leq k/2 +2$ holds~\cite{DBLP:journals/ejc/ChepoiD00}. For every graph $G$ embeddable non-contractively into a (weighted) tree with multiplication distortion $\alpha$, $\Delta_s(G) \leq 3\alpha$ holds~\cite{ChepoiDNRV12}. See Section \ref{sec:td} for more on this topic.

\begin{table}
\footnotesize
\begin{center}
\begin{tabular}{ | c | c | c | c|}
    \hline
  Graph     & n=     & m=    & $ \delta(G) $  \\
  $G=(V,E)$ & $|V|$  & $|E|$ &  \\ \hline\hline
  PPI & 1458 & 1948 & 3.5 \\ \hline
  Yeast & 2224 & 6609 & 2.5\\ \hline
  DutchElite  & 3621 & 4311 & 4\\ \hline
  EPA  &4253 &8953 & 2.5\\ \hline
  EVA  &4475 &4664 &1\\ \hline
  California & 5925 & 15770 & 3 \\ \hline
  Erd\"os  & 6927 & 11850 & 2 \\ \hline
  Routeview  & 10515 & 21455 & 2.5\\ \hline
  Homo  release 3.2.99  & 16711& 115406 & 2\\ \hline
  AS\_Caida\_20071105 & 26475 & 53381 & 2.5 \\ \hline
  Dimes 3/2010  & 26424 & 90267 &2 \\ \hline
  Aqualab 12/2007- 09/2008  & 31845 & 143383 &  2 \\ \hline
  AS\_Caida\_20120601 & 41203 & 121309 &  2 \\
  \hline
\end{tabular}
\end{center}
\caption{$\delta$-hyperbolicity of the graph datasets.}
\label{tab:hyper}
\vspace*{-0.9cm}\end{table}

\section{Hyperbolicity} \label{sec:hyperbol}

$\delta$-Hyperbolic metric spaces have been defined by M. Gromov
\cite{Gromov87} in 1987 via a simple 4-point condition: for any four
points $u,v,w,x$, the two larger of the distance sums
$d(u,v)+d(w,x), d(u,w)+d(v,x), d(u,x)+d(v,w)$ differ by at most
$2\delta$. They play an important role in geometric group theory,
geometry of negatively curved spaces, and have recently become of
interest in several domains of computer science, including
algorithms and networking. For example, (a) it has been shown
empirically in~\cite{DBLP:journals/ton/ShavittT08} (see also~\cite{DBLP:conf/podc/AbrahamBKMRT07}) that the Internet
topology embeds with better accuracy into a hyperbolic space than
into an Euclidean space of comparable dimension, (b) every connected
finite graph has an embedding in the hyperbolic plane so that the
greedy routing based on the virtual coordinates obtained from this
embedding is guaranteed to work (see~\cite{DBLP:conf/infocom/Kleinberg07}). A connected
graph $G=(V,E)$ equipped with standard graph metric $d_G$ is
$\delta$-{\it hyperbolic} if the metric space $(V,d_G)$ is
$\delta$-hyperbolic.


More formally, let $G$ be a graph and $u, v, w$ and $x$ be its four vertices. Denote by $S_1, S_2, S_3$ the three distance sums, $d_G(u,v)+d_G(w,x)$, $d_G(u,w)+d_G(v,x)$ and $d_G(u,x) + d_G(v,w)$ sorted in non-decreasing order $S_1 \leq S_2 \leq S_3$. Define the {\em hyperbolicity of a quadruplet} $u,v,w,x$ as $\delta(u,v,w,x)=\frac{S_3-S_2}{2}$. Then the {\em hyperbolicity $\delta(G)$ of a graph} $G$ is the maximum hyperbolicity over all possible quadruplets
of $G$, i.e., \[\delta(G)=\max_{u,v,w,x\in V}\delta(u,v,w,x).\]

$\delta$-Hyperbolicity measures the local deviation of a metric from a tree metric; a metric is a tree metric if and only if it has hyperbolicity $0$. Note that chordal graphs, mentioned in Section \ref{sec:layer-partit}, have hyperbolicity at most $1$~\cite{uea21813}, while $k$-chordal graphs have hyperbolicity at most $k/4$~\cite{DBLP:journals/combinatorics/WuZ11}.


In Table~\ref{tab:hyper}, we show the hyperbolicities of most of our graph datasets. The computation of hyperbolicities is a costly operation. We did not compute it for only three very large graph datasets since it would take very long time to calculate. The best known algorithm to calculate hyperbolicity has time complexity of $O(n^{3.69})$, where $n$ is the number of vertices in the graph; it was proposed in~\cite{FournierHyper} and involves matrix multiplications. This algorithm still takes long running time for large graphs and is hard to implement. Authors of~\cite{FournierHyper} also propose a $2$-approximation algorithm for calculating hyperbolicity that runs in $O(n^{2.69})$ time and a $2\log_2 n$-approximation algorithm that runs in $O(n^2)$ time. 
In our computations, we used the naive algorithm which calculates the exact hyperbolicity of a given graph in $O(n^{4})$ time via calculating the hyperbolicities of its quadruplets. It is easy to show that the hyperbolicity of a graph is  realized on its biconnected component. Thus, for very large graphs, we needed to check hyperbolicities only for quadruplets coming from the same biconnected component. Additionally, we used an algorithm by Cohen et. el. from~\cite{cohenHyper} which has $O(n^4)$ time complexity but performs well in practice as it prunes the search space of quadruplets.


It turns out that most of the quadruplets in our datasets have small $\delta$ values (see Table~\ref{tab:hyper-dist}). For example, more than $96\%$ of vertex quadruplets in EVA and Erd\"os datasets have $\delta$ values equal to $0$. For the remaining graph datasets in Table~\ref{tab:hyper-dist}, more than $96\%$ of the quadruplets have $\delta \leq 1$, indicating that all of those graphs are metrically very close to trees.

\begin{table}
\footnotesize
\begin{center}
\begin{tabular}{ | c | c |c | c | c | c | c | c |}
    \hline
  \backslashbox{~~$\delta$}{Graph} & PPI & Yeast  & DucthElite & EPA & EVA & California & Erd\"os \\ \hline\hline
  0   & 0.4831 & 0.487015& 0.54122195 &0.5778 &0.9973 &0.49057007 &0.96694 \\ \hline
  0.5 & 0.3634 & 0.450362& 0 &0.3655 &0.0007 &0.41052969 &0.03278 \\ \hline
  1   & 0.1336 & 0.060844 &0.42201697 &0.0552 & 0.0020 & 0.09527387 &0.00028 \\ \hline
  1.5 & 0.0179 & 0.001762 &0 &0.0015 & -- & 0.00344690 & 6.80E-08 \\ \hline
  2   & 0.0019 & 0.000017 &0.03642388 &2.09E-05 & -- &0.00017945 & 3.64E-11\\ \hline
  2.5 & 3.55E-05& 2.4641E-09& 0 &1.37E-10 & -- & 0.00000001 & -- \\ \hline
  3   & 1.65E-06 & -- & 0.00033717& -- & -- &1.88E-11 & --\\ \hline
  3.5 & 3.79E-09&-- & 0 & --& --&-- &-- \\ \hline
  4   & -- &-- & 0.00000004 & --&-- & --& --\\ \hline\hline
  \% $\leq 1$ & 98.01 & 99.8221 & 96.323891 & 99.84 & 100 & 99.637364 & 99.99999 \\
  \hline
\end{tabular}
\end{center}
\caption{ Relative frequency of $\delta$-hyperbolicity of quadruplets in our graph datasets that have less than 10K vertices.}
\label{tab:hyper-dist}
\vspace*{-0.5cm}\end{table}

In the remaining part of this section, we discuss the theoretical relations between parameters $\delta(G)$ and $\Delta_s(G)$ of a graph.
In~\cite{DBLP:conf/compgeom/ChepoiDEHV08}, the following inequality was proven.

\commentout{
\begin{proposition}[\cite{alonso90,Bandelt08,Gromov87}]
\label{prop:geoTriHyper}
Geodesic triangles of $\delta$-hyperbolic space are $4\delta$-thin.
\end{proposition}
\begin{proposition}[\cite{DBLP:conf/compgeom/ChepoiDEHV08}]
\label{prop:geoTriDelta}
For a graph $G$ with $\delta$-thin geodesic triangles and $n$ vertices, then $\Delta_s \leq 4+3\delta+2\delta \log n$.
\end{proposition}
}

\begin{proposition}[\cite{DBLP:conf/compgeom/ChepoiDEHV08}]
\label{prop:diam-leq-deltaHyper}
For every $n$-vertex graph $G$ and any vertex $s$ of $G$, $$\Delta_s(G) \leq 4+12\delta(G)+8\delta(G) \log_2 n.$$
\end{proposition}

Here we complement that inequality by showing that the hyperbolicity of a graph is at most $\Delta_s(G)$.
\begin{proposition} \label{prop:hyp<Cl-diam} 
\label{lem:diam-hyper}
For every $n$-vertex graph $G$ and any vertex $s$ of $G$, $$\delta(G) \leq \Delta_s(G).$$
\end{proposition}

\begin{proof}
Let $\mathcal{LP}(G,s)$ be a layering partition of $G$ and $\Gamma(G,s)$ be the corresponding layering tree (consult Fig. \ref{fig:layering-partition}). From construction of $\mathcal{LP}(G,s)$ and $\Gamma(G,s)$, every cluster $C$ of $\mathcal{LP}(G,s)$ separates in $G$ any two vertices belonging to nodes (clusters) of different subtrees of the forest obtained from $\Gamma(G,s)$ by removing node $C$. Note that every vertex of $G$ belongs to exactly one node (cluster) of the layering tree $\Gamma(G,s)$.

Consider an arbitrary quadruplet $x,y,z,w$ of vertices of $G$. Let $X,Y,Z,W$ be the four nodes in $\Gamma(G,s)$ (i.e., four clusters in $\mathcal{LP}(G,s)$)  containing vertices $x,y,z,w$, respectively.  In the tree $\Gamma(G,s)$, consider a median node $M$ of nodes $X,Y,Z,W$, i.e., a node $M$ removing of which from $\Gamma(G,s)$ leaves no connected subtree with more that two nodes from $\{X,Y,Z,W\}$. As a consequence, any connected component of graph $G[V\setminus M]$ (the graph obtained from $G$ by removing vertices of $M$) cannot have more than $2$ vertices out of $\{x,y,z,w\}$. Thus, $M$ separates at least $4$ pairs out of the $6$ possible pairs formed by vertices $x,y,z,w$.  Assume, without loss of generality, that $M$ separates in $G$ vertices $x$ and $y$ from vertices $z$ and $w$.
See Fig.~\ref{fig:diam-hyper} for an illustration.

\begin{figure}\vspace*{-.5cm}
        \centering
        \begin{subfigure}[b]{0.35\textwidth}
                \includegraphics[width=\textwidth]{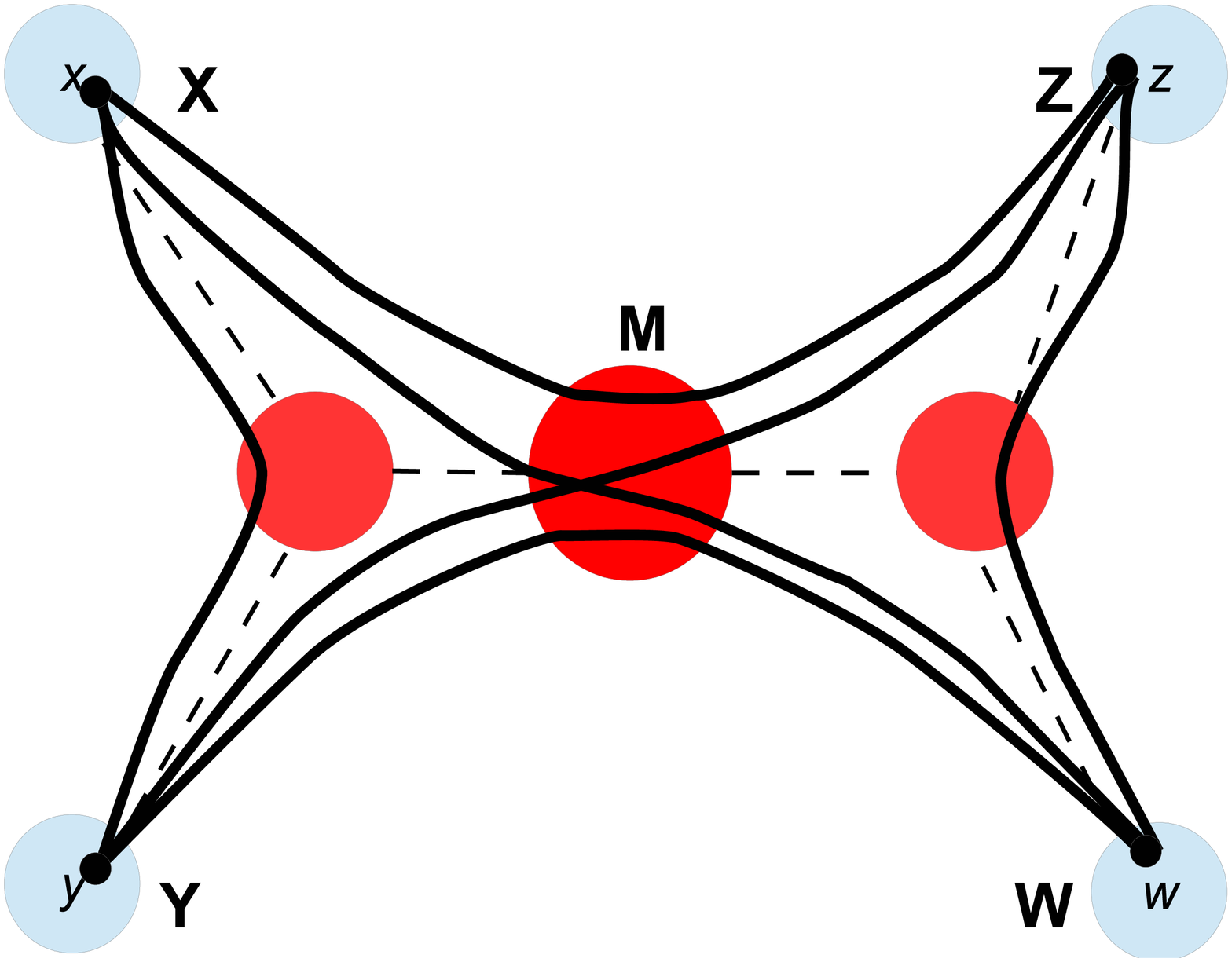}
                \caption{$M$ is a median node for $X,Y,Z,W$ in $\Gamma(G,s)$.} 
                \label{fig:diam-hyper1}
        \end{subfigure}%
          \qquad
        \begin{subfigure}[b]{0.35\textwidth}
                \includegraphics[width=\textwidth]{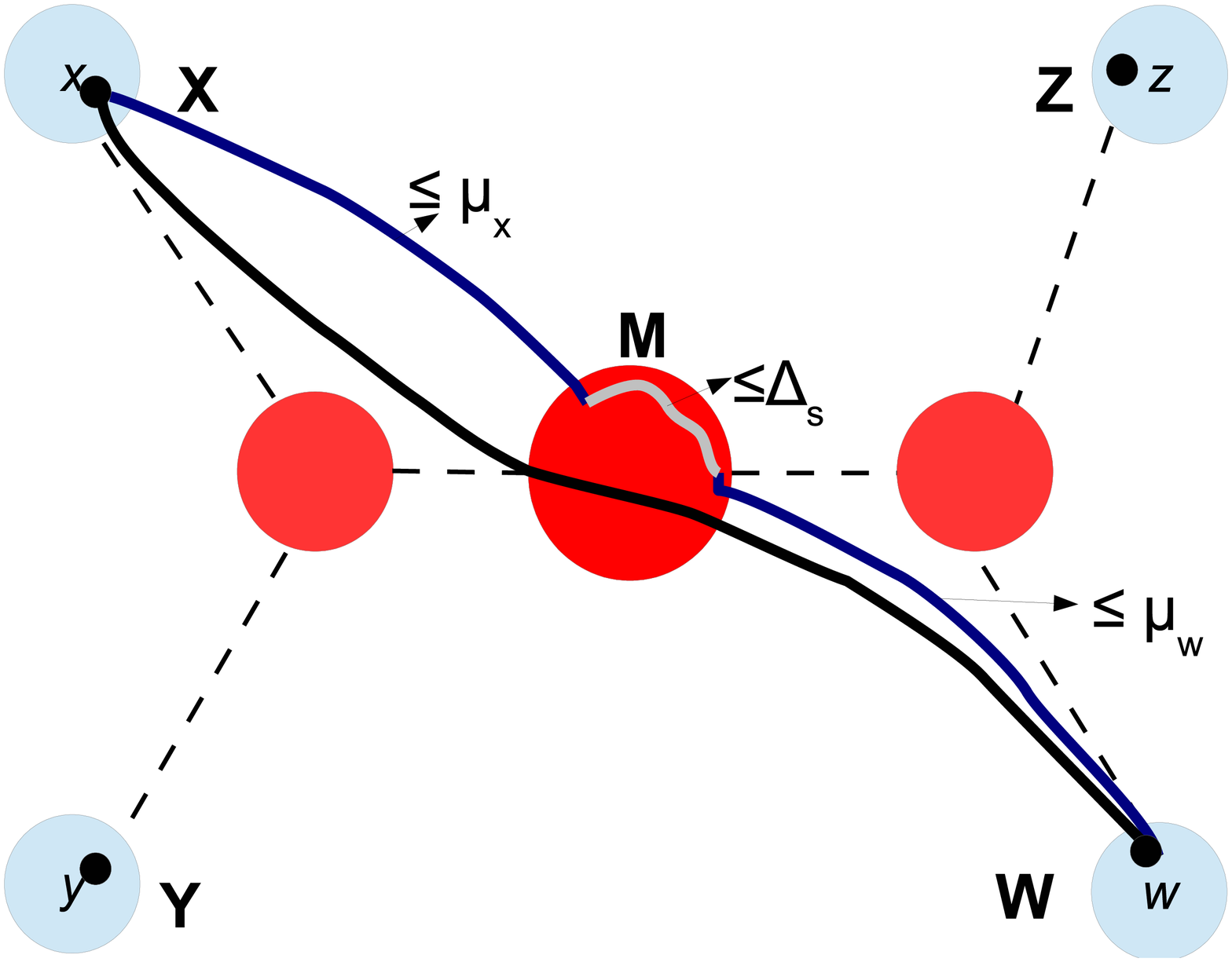}
                \caption{$M$ separates in $G$ vertices $x$ and $y$ from vertices $z$ and $w$.}
                \label{fig:diam-hyper2}
        \end{subfigure}
         \caption{Illustration to the proof of  Proposition~\ref{lem:diam-hyper}.}\label{fig:diam-hyper}
\vspace*{-.3cm}
\end{figure}

Let $\mu_a$ be the distance from $a \in \{x,y,z,w\}$ to its closest vertex in $M$.
Let $a,b$ be a pair of vertices from $\{x,y,z,w\}$.
If the vertices $a,b$ belong to different components of $G[V \setminus M]$, then $M$ separates $a$ from $b$ and therefore $\mu_a+ \mu_b \leq d_G(a,b)$.
Since $M$ separates in $G$ vertices $x$ and $y$ from vertices $z$ and $w$, we get $d_G(x,z)+d_G(y,w)\geq \mu_x+\mu_y+\mu_z+\mu_w$ and $d_G(x,w)+d_G(y,z)\geq \mu_x+\mu_y+\mu_z+\mu_w$.
On the other hand, all three sums $d_G(x,z)+d_G(y,w)$, $d_G(x,w)+d_G(y,z)$ and  $d_G(x,y)+d_G(z,w)$ are less than or equal to $\mu_x+\mu_y+\mu_z+\mu_w+2\Delta_s(G)$, since, by the triangle inequality, $d_G(a,b) \leq \mu_a + \mu_b +\Delta_s(G)$ for every $a,b \in \{x,y,z,w\}$.
\commentout{
If the shortest path between a pair $a,b \in \{x,y,z,t\}$ passes through $M$. Let $a',b'$ be the closest vertices in $M$ to $a,b$ respectively. Then, by triangle inequality, we have $d_G(a,b) \leq d_G(a,a')+ d_G(a',b')+d_G(b,b') \leq \mu_a + \mu_b +\Delta_s$
  Otherwise, the shortest path lies entirely in one component of $G \backslash M$. Let $M'$ be the median node between the two nodes in $X,Y,Z,T$ containing $a,b$ in $\Gamma$. Let $\mu'_a$ and $\mu'_b$ be the distance from $a,b$ respectively, to the closest vertex in $M'$. By triangle inequality, we have $d_G(a,b) \leq \mu'_a + \mu'_b +\Delta_s$. Since $M'$ is closer to the nodes containing $a,b$ than $M$, then we have $d_G(a,b) \leq \mu'_a + \mu'_b +\Delta_s \leq \mu_a + \mu_b +\Delta_s$. Therefore, any sum is less than or equal $\mu +2\Delta_s$.
}
Now, since the two larger distance sums are between $\mu$ and $\mu +2\Delta_s(G)$, where $\mu:=\mu_x+\mu_y+\mu_z+\mu_w$, we conclude that the difference between the two larger distance sums is at most $2\Delta_s(G)$. Thus, necessarily $\delta(G)\leq \Delta_s(G)$.\qed
\commentout{
First, we show that $\delta(G) \leq \Delta_s(G)$. Let $\mathcal{LP}(s)$ be a layering partition of the graph $G$ with layering tree $\Gamma$. For any quadruplet of vertices $x,y,z,t$, each one of them belong to exactly one cluster (node) of  the layering partition tree $\Gamma$. Let $X,Y,Z,T$ be the four clusters containing the  vertices $x,y,z,t$, respectively. Consider the nearest common ancestor clusters of any three clusters $X,Y,Z,T$ in $\Gamma$. Among these common ancestors, pick the one having the highest depth, name it $C$. Let $\beta_a$ be the distance from $a \in \{x,y,z,t\}$ to it's closest vertex in $C$. Let $P(a,b)$ be a shortest path between $a$ and $b$ passing through $C$. Then we have $\beta_a + \beta_b \leq d_G(a,b) \leq \beta_a + \beta_b + \Delta_s $. Since $C$ is a separator in $G$ between any pair $a,b$ in $x,y,x,t$ except maybe for one pair. Then, at least 5 out of the 6 shortest paths connecting these pair in $x,y,z,t$ passes through $C$.  This guarantees that the largest two sums among $d_G(x,y)+d_G(z,t)$, $d_G(x,z)+d_G(y,t)$ and $d_G(x,t)+d_G(x,z)$ greater than or qual $\beta_x+\beta_y+\beta_z+\beta_t=\beta$ and smaller than or equal $\beta +2\Delta_s$. Thus, the difference between the largest two sums ($S_1 -S_2$) is less than or equal $2\Delta_s$ yielding $\frac{S_1-S_2}{2} \leq \Delta_s$.
Second, propositions ~\ref{prop:diam-leq-deltaHyper} shows that $\Delta_s \leq 4+12\delta+8\delta \log_2 n$.
}
\end{proof}

Combining Proposition~\ref{prop:diam-leq-deltaHyper} with Proposition \ref{lem:cluster-diam}, one obtains also the following interesting result relating the hyperbolicity of a graph $G$ with additive distortion of embedding of $G$ to its canonic tree $H$.
 \begin{proposition}[\cite{DBLP:conf/compgeom/ChepoiDEHV08}]
\label{lem:treeH-hyper}
For any graph $G=(V,E)$ and its canonic tree $H=(V,F)$ the following is true:
\[\forall u,v \in V, ~~d_H(u,v)-2 \leq d_G(u,v) \leq d_H(u,v)+ O(\delta(G)\log n). \]
\end{proposition}

Since a canonic tree $H$ is constructible in linear time for a graph $G$, by Proposition \ref{lem:treeH-hyper}, the distances in $n$-vertex $\delta$-hyperbolic graphs can efficiently be approximated within an additive error of $O(\delta \log n)$ by a tree metric and this approximation is
sharp (see~\cite{Gromov87,GhHa} and~\cite{DBLP:conf/compgeom/ChepoiDEHV08,GaLy}).

Graphs and general geodesic spaces with small hyperbolicities have many other algorithmic advantages. They allow efficient approximate solutions for a number of optimization problems. For example, Krauthgamer and Lee~\cite{KrLe} presented a PTAS for the Traveling Salesman Problem when the set of cities lie in a hyperbolic metric space. Chepoi and Estellon~\cite{DBLP:conf/approx/ChepoiE07} established a relationship between the minimum number of balls of
radius $r+2\delta$ covering a finite subset $S$ of a
$\delta$-hyperbolic geodesic space and the size of the maximum
$r$-packing of $S$ and showed how to compute such coverings and
packings in polynomial time. Chepoi et al. gave in~\cite{DBLP:conf/compgeom/ChepoiDEHV08}
efficient algorithms for fast and accurate estimations of diameters
and radii of $\delta$-hyperbolic geodesic spaces and graphs. Additionally, Chepoi et al. showed in~\cite{ChDrEsRout} that every $n$-vertex $\delta$-hyperbolic
graph has an additive $O(\delta \log n)$-spanner with at most
$O(\delta n)$ edges and enjoys an $O(\delta\log
n)$-additive routing labeling scheme with $O(\delta\log^2n)$ bit
labels and $O(\log\delta)$ time routing protocol.  We 
elaborate more on these results in Section \ref{appl}.

\section{Tree-Distortion} \label{sec:td}
The problem of approximating a given graph metric by a ``simpler'' metric is well 
motivated from several different perspectives. A particularly simple
metric of choice, also favored from the algorithmic point of view,
is a tree metric, i.e., a metric arising from shortest path distance
on a tree containing the given points. In recent years, a number of authors considered
problems of minimum distortion embeddings of graphs into trees (see~\cite{AgBaFaNaPa,BaDeHaSiZa,BaInSi,ChepoiDNRV12}), most popular among them being a
non-contractive embedding with minimum multiplicative distortion.

Let $G=(V,E)$ be a graph. The (multiplicative) \emph{tree-distortion} $td(G)$
of $G$ is the smallest integer $\alpha$ such that $G$ admits a tree  (possibly weighted and with Steiner points)
with $$\forall u,v \in V, ~~d_G(u,v) \leq d_T(u,v) \leq \alpha~ d_G(u,v).$$ The problem of finding, for a given graph $G$, a tree $T=(V\cup S, F)$ satisfying  $d_G(u,v) \leq d_T(u,v) \leq td(G) d_G(u,v)$, for all $u,v \in V$, is known as the {\em problem of minimum distortion non-contractive embedding of graphs into trees}. In a non-contractive embedding, the distance in the tree must always be larger that or equal to the distance in the graph, i.e., the tree distances ``dominate'' the graph distances.

It is known that this problem is NP-hard, and even more, the hardness result of~\cite{AgBaFaNaPa}
implies that it is NP-hard to approximate $td(G)$ better than $\gamma$, for some small constant $\gamma$.
The best known 6-approximation algorithm using layering partition technique was recently given in~\cite{ChepoiDNRV12}. It improves the previously known
100-approximation algorithm from~\cite{BaInSi} and 27-approximation algorithm from~\cite{BaDeHaSiZa}. Below we will provide a short description of the method of ~\cite{ChepoiDNRV12}.

The following proposition establishes relationship between the tree-distortion and the cluster-diameter of a graph.
\begin{proposition}[\cite{ChepoiDNRV12}]
\label{lem:td-cluster-diam}
For every graph $G$ and any its vertex $s$, $\Delta_s(G)/3\leq td(G)\leq 2\Delta_s(G)+2.$
\end{proposition}

Proposition~\ref{lem:td-cluster-diam} shows that the cluster-diameter $\Delta_s(G)$ of a layering partition of a graph $G$ linearly bounds the tree-distortion $td(G)$ of $G$.

Combining Proposition~\ref{lem:td-cluster-diam} and Proposition~\ref{lem:cluster-diam}, the following result is obtained.
\begin{proposition}[\cite{ChepoiDNRV12}]
\label{lem:treeH-td}
For any graph $G=(V,E)$ and its canonic tree $H=(V,F)$ the following is true:
\[\forall u,v \in V, ~~d_H(u,v)-2 \leq d_G(u,v) \leq d_H(u,v)+ 3~td(G). \]
\end{proposition}

Surprisingly, a multiplicative  distortion turned into an additive distortion. Furthermore, while a tree $T=(V\cup S, F)$ satisfying  $d_G(u,v) \leq d_T(u,v) \leq td(G) d_G(u,v)$, for all $u,v \in V$, is NP-hard to find, a canonic tree $H$ of $G$ can be constructed in $O(m)$ time (where $m=|E|$).

By assigning proper weights to edges of a canonic tree $H$ or adding at most $n=|V|$ new Steiner points to $H$, the authors of~\cite{ChepoiDNRV12} achieve a good non-contractive embedding of a graph $G$ into a tree. Recall that a canonic tree $H=(V,F)$ of $G=(V,E)$ is constructed in the following way: identify for each cluster $C=L^i_j \in \mathcal{LP}(G,s)$ of a layering partition $\mathcal{LP}(G,s)$ of $G$ an arbitrary vertex $x_C \in L_{i-1}$  which has a neighbor in $C = L^i_j$ and make $x_C$ adjacent in $H$ with all vertices $v\in C$ (see Fig. \ref{fig:treeHHp}). Note that $H$ is an unweighted tree, without any Steiner points, and  resembles a BFS-tree of $G$. Two other trees for $G$ are constructed as follows.

\commentout{
\begin{enumerate}
  \item
  \item $H_{\ell}$: a weighted tree without Steiner nodes that produce a non-contractive embedding
  \item $H'_{\ell}$: a weighted tree with Steiner nodes that produce non-contractive embedding
\end{enumerate}
}

\textbf{Tree} $\mathbf{H_{\ell}:}$ Tree $H_{\ell}=(V,F,\ell)$ is obtained from $H$ by assigning uniformly the weight $\ell=\max\{d_G(u,v): uv \mbox{ is an edge of } H\}$ to all edges of $H$. So, $H_\ell$ is a uniformly weighted tree without Steiner points. It turns out that $G$ embeds in tree $H_\ell$ non-contractively. Note that,  although the topology of the tree $H_\ell$ can be determined in $O(m)$ time ($H_\ell$ is isomorphic to $H$), computation of the weight $\ell$ requires $O(n m)$ time. Thus, the tree $H_\ell$ is constructible in $O(n m)$ total time. See Fig.~\ref{fig:treeHHp} for an illustration.

\textbf{Tree} $\mathbf{H'_{\ell}:}$  Tree $H'_{\ell}=(V\cup S,F',\ell)$ is obtained from $H$ by first introducing one Steiner point $p_C$  for each cluster $C := L^i_j$  and adding an edge between each vertex of $C$ and $p_C$ and an edge between $p_C$ and the support vertex $x_C$ for $C$, and then by assigning uniformly the weight $\ell=\frac{1}{2}\max\{\Delta_s(G),\max\{d_G(u,v): uv \mbox{ is an edge of } H\}\}$ to all edges of the obtained tree. So, $H'_\ell$ is a uniformly weighted tree with at most $O(n)$ Steiner points. Again, $G$ embeds into tree $H'_\ell$ non-contractively and $H'_\ell$ can be obtained in $O(n m)$ total time.
See Fig.~\ref{fig:treeHpell} for an illustration.

\begin{figure}\vspace*{-1.7cm}
        \centering
        \begin{subfigure}[b]{0.45\textwidth}
                \includegraphics[width=\textwidth]{figures/layering3.eps}
                \vspace*{-.3cm}
                \caption{Topology of trees $H$ and $H_{\ell}$.}
                \label{fig:treeHHp}
        \end{subfigure}
          ~
        \begin{subfigure}[b]{0.45\textwidth}
                \includegraphics[width=\textwidth]{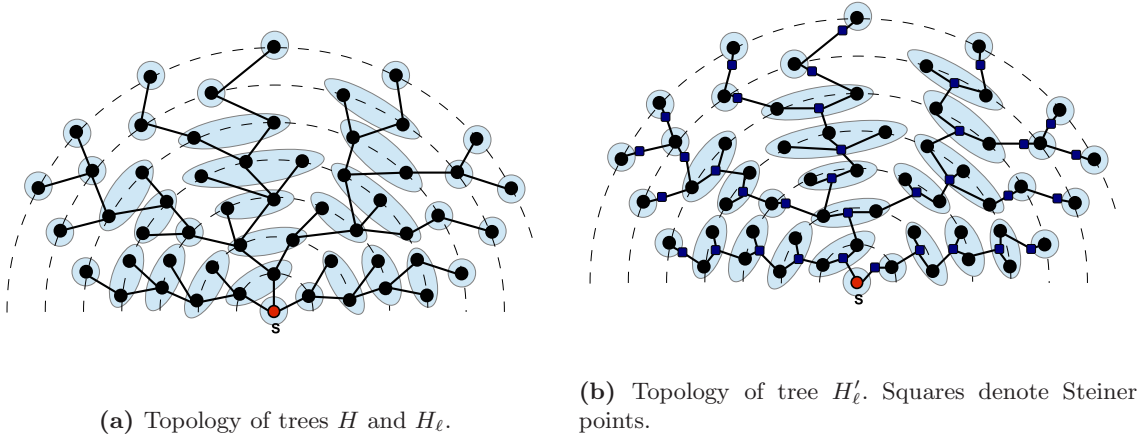}
                \vspace*{-.3cm}
                \caption{Topology of tree $H'_{\ell}$. Squares denote Steiner points.}
                \label{fig:treeHpell}
        \end{subfigure}
         \caption{Embedding into trees $H, H_{\ell}$ and $H'_{\ell}$.}
         \label{fig:trees}
\end{figure}

Constructed trees have the following distance properties (for comparison reasons, we include also the results for $H$ mentioned earlier).
\begin{proposition}[\cite{ChepoiDNRV12}]
\label{tree-properties}
Let  $G=(V,E)$ be a graph, $s$ be its arbitrary vertex, $\alpha=td(G)$, $\Delta_s=\Delta_s(G)$, and $H$, $H_{\ell}$, $H'_{\ell}$ be trees as described above.
Then, for any two vertices $x$ and $y$ of $G$, the following is true:
\[ d_H(x,y)-2 \leq d_G(x,y) \leq d_H(x,y)+\Delta_s,\]
\[ d_H(x,y)-2 \leq d_G(x,y) \leq d_H(x,y)+ 3\alpha,\]

\[d_G(x,y)\leq d_{H_{\ell}}(x,y)  \leq (\Delta_s+1)(d_G(x,y)+2),\]
\[d_G(x,y)\leq d_{H_{\ell}}(x,y)  \leq  \max\{3\alpha-1,2\alpha+1\} \left(d_G(x,y)+2\right),\]

\[d_G(x,y)\leq d_{H'_{\ell}}(x,y)  \leq (\Delta_s+1)(d_G(x,y)+1),\]
\[d_G(x,y)\leq d_{H'_{\ell}}(x,y)  \leq  3\alpha (d_G(x,y)+1).\]
\end{proposition}

As pointed out in~\cite{ChepoiDNRV12}, tree $H'_{\ell}$ provides a $6$-approximate solution to the problem of minimum distortion non-contractive embedding of graph into tree.

In our empirical study, we analyze embeddings of our graph datasets into each of these three trees and measure how close these graph datasets resemble a tree from this prospective. We compute the following measures:

\vspace*{-.3cm}

\begin{enumerate}
  \item[-] {\em maximum distortion right}  $:= \max\{\frac{d_T(u,v)}{d_G(u,v)}: u,v \in V, ~d_T(u,v) > d_G(u,v)>0\}$;
  \item[-] {\em maximum distortion left} $:= \max\{\frac{d_G(u,v)}{d_T(u,v)}: u,v \in V, ~d_G(u,v) > d_T(u,v)>0\}$;
  \item[-] {\em average distortion right}  $:= \mbox{avg}\{\frac{d_T(u,v)}{d_G(u,v)}: u,v \in V, ~d_T(u,v) > d_G(u,v)>0\}$;
  \item[-] {\em average distortion left} $:= \mbox{avg}\{\frac{d_G(u,v)}{d_T(u,v)}: u,v \in V, ~d_G(u,v) > d_T(u,v)>0\}$;

  \item[-] {\em average relative distortion} $:=\mbox{avg}\{\frac{|d_T(u,v)-d_G(u,v)|}{d_G(u,v)} : u,v \in V\}$;
  \item[-] {\em distance-weighted average distortion} $:=\frac{1}{\Sigma_{u,v \in V}d_G(u,v)}\Sigma_{u,v \in V}(d_G(u,v)\cdot \frac{d_T(u,v)}{d_G(u,v)})=\frac{\Sigma_{u,v \in V}d_T(u,v)}{\Sigma_{u,v \in V}d_G(u,v)}$.
\end{enumerate}

A pair of distinct vertices $u,v$ of $G=(V,E)$ we call a {\em right pair} with respect to tree $H=(V,F)$ if $d_G(u,v) < d_H(u,v)$. If $d_H(u,v) < d_G(u,v)$ then they are called a {\em left pair}. Note that $G$ has no left pairs with respect to trees $H_{\ell}$ and $H'_{\ell}$, hence.  in case of trees $H_{\ell}$ and $H'_{\ell}$, we talk only about maximum distortion, average distortion, average relative distortion and distance-weighted average distortion.
Distance-weighted average distortion is used in literature  when distortion of distant pairs of vertices is more important than that of close pairs, as it gives larger weight values to distortion of distant pairs (see~\cite{DBLP:journals/corr/Kao13}). Clearly, any tree graph would have maximum distortion, average relative distortion and distance-weighted average distortion equal to 1, 0 and 1, respectively.

Tables~\ref{tab:treeH} and \ref{tab:treeHl}  show the results of embedding our graph datasets into trees $H,$ $H_{\ell}$ and $H'_{\ell}$, respectively. It turns out that most of the datasets embed into tree $H$ with average distortion (right or left, right being usually better) between $1$ and $1.5$. Also, many pairs of vertices enjoy exact embedding to tree $H$; they preserve their original graph distances (for example, around $88\%$ of the pairs in Erd\"os dataset, $72\%$ of pairs in Homo release 3.2.99, $57\%$ in AS\_Caida\_20120601 preserve their original graph distances). Comparing the results of non-contractive embeddings to trees $H_{\ell}$ and $H'_{\ell}$, we observe that max distortions are slightly improved in $H'_{\ell}$ over distortions in $H_{\ell}$, but average distortions are very much comparable. Furthermore, distance-weighted average distortions are better in $H_{\ell}$ than in $H'_{\ell}$. This confirms the Gupta's claim in~\cite{DBLP:conf/soda/Gupta01} that the Steiner points do not really help. 


\begin{table}
\footnotesize
\begin{center}
\hspace*{-50mm}
\begin{tabular}{ | c |      p{1.5cm} | p{1.0cm} | p{1.1cm} | p{1.5cm}  | p{1.0cm}  | p{1.1cm} | p{1.2cm} | p{1.5cm} |p{1.5cm} |}
    \hline
  Graph                     & average distor\-tion left
                                       & max distor\-tion left
                                             & \% of~left pairs (round.)
                                                      & average distor\-tion right
                                                                & max distor\-tion right
                                                                    & \% of~right pairs (round.)
                                                                            & \% of~pairs $d_T=d_G$ (round.)
                                                                                     & average relative distortion
                                                                                                 & distance-weighted average distortion\\ \hline\hline
  PPI                       &~ 1.50159 &~~7  &~ 70.5   &~1.34140 &~3 &~~9.1  &~20.4   &~0.24669   &~0.790311\\ \hline
  Yeast                     &~ 1.48714 &~~5  &~ 56.3   &~1.38989 &~3 &~12.2  &~31.5   &~0.219268  &~0.850311\\ \hline
  DutchElite                &~ 1.54045 &~~7  &~ 73.0   &~1.41254 &~3 &~~3.9  &~23.1   &~0.252341  &~0.760714\\ \hline
  EPA                       &~ 1.50416 &~~5  &~ 44.66  &~1.38107 &~3 &~10.47 &~44.87  &~0.178557  &~0.878082\\ \hline
  EVA                       &~ 1.29905 &~~6  &~ 32.31  &~1.27780 &~3 &~14.77 &~52.92  &~0.110271  &~0.951626 \\ \hline
  California                &~ 1.52477 &~~5  &~ 61.82  &~1.37071 &~3 &~~7.92 &~30.25  &~0.227176  &~0.810647 \\ \hline
  Erd\"os                   &~ 1.35242 &~~3  &~ ~2.75  &~1.41097 &~3 &~~8.91 &~88.34  &~0.0437277 &~1.02241\\ \hline
  Routeview                 &~ 1.40636 &~~4  &~ 24.39  &~1.41413 &~3 &~33.34 &~42.28  &~0.205375  &~1.03343\\ \hline
  Homo  release 3.2.99      &~ 1.533   &~~4  &~ ~2.83   &~1.67827 &~3 &~25.16 &~72.01  &~0.180092  &~1.13402\\ \hline
  AS\_Caida\_20071105       &~ 1.48085 &~~4  &~ 21.43  &~1.35730 &~3 &~35.42 &~43.15  &~0.192302  &~1.02943 \\ \hline
  Dimes 3/2010              &~ 1.53666 &~~3  &~ ~5.74   &~1.37247 &~3 &~44.42 &~49.84  &~0.184767  &~1.12555\\ \hline
  Aqualab 12/2007- 09/2008  &~ 1.42269 &~~4  &~ 31.71  &~1.41923 &~3 &~35.75 &~32.54  &~0.241815  &~1.03194 \\ \hline
  AS\_Caida\_20120601       &~ 1.34538 &~~4  &~ 22.42  &~1.40429 &~3 &~20.43 &~57.15  &~0.138869  &~1.0068 \\ \hline
  itdk0304                  &~ 1.60077 &~~8  &~ 94.85  &~1.26367 &~3 &~~0.55 &~~4.60  &~0.331656  &~0.673012 \\ \hline
  DBLB-coauth               &~ 1.77416 &~~9  &~ 95.82  &~1.24977 &~3 &~~0.59 &~~3.59  &~0.383101  &~0.615328\\ \hline
  Amazon                    &~ 2.48301 &~19 &~ 99.17  &~1.20027 &~3 &~~0.20 &~~0.63  &~0.536656  &~0.536656\\
  \hline
\end{tabular}\hspace*{-50mm}
\end{center}
\caption{Distortion results of embedding datasets into a canonic tree $H$.}
\label{tab:treeH}
\vspace*{-0.5cm}\end{table}

As tree $H'_{\ell}$ provides a $6$-approximate solution to the problem of minimum distortion non-contractive embedding of graph into tree, dividing by 6 the max
distortion values in Table \ref{tab:treeHl} for tree $H'_{\ell}$, we obtain a lower bound on $td(G)$ for each graph dataset $G$. For example, $td(G)$ is at lest 4/3 for Erd\"os and Dimes 3/2010, at least 5/3 for  Homo  release 3.2.99, at least 2 for Yeast, EPA, Routeview, AS\_Caida\_20071105, Aqualab 12/2007-09/2008 and  AS\_Caida\_20120601, at least 8/3 for PPI and California, at least 10/3 for DutchElite, at least 3 for EVA, at least 11/3 for  itdk0304 and  DBLB-coauth, at least 7 for  Amazon.

\begin{table}
\footnotesize
\begin{center}
\begin{tabular}{  p{4.0cm}   p{4.0cm}   c  } 
 ~~~~~~~~~~~~~~~~~~~~~~~~~~~~~~~ & ~{\bf\large tree $H_{\ell}$}~   & ~~~~~~~~~~~~~~{\bf\large tree $H'_{\ell}$}~ \\
\end{tabular}

\hspace*{-50mm}
\begin{tabular}{ | c || p{1.3cm} | p{1.2cm} | p{1.3cm}| p{1.3cm}|| p{1.3cm} | p{1.2cm} | p{1.3cm} | p{1.3cm}|} \hline
  Graph & average distor\-tion & max distor\-tion & average relative distor\-tion & distance-weighted average distor\-tion & average distor\-tion & max distor\-tion & average relative distor\-tion & distance-weighted average distor\-tion \\  \hline\hline
  PPI                       &~5.70566 &~ 21 &~4.70566 &~5.53218 &~5.29652 &~ 16 &~4.29652 &~5.2027  \\ \hline
  Yeast                     &~4.37781 &~ 15 &~3.37781 &~4.25155 &~3.79318 &~ 12 &~2.79318 &~3.74159 \\ \hline
  DutchElite                &~5.45299 &~ 21 &~4.45299 &~5.325   &~6.53269 &~ 20 &~5.53269 &~6.4574  \\ \hline
  EPA                       &~4.50619 &~ 15 &~3.50619 &~4.39041 &~4.06901 &~ 12 &~3.06901 &~3.99447 \\ \hline
  EVA                       &~5.83084 &~ 18 &~4.83084 &~5.70976 &~7.77752 &~ 18 &~6.77752 &~7.65544 \\ \hline
  California                &~4.15785 &~ 15 &~3.15785 &~4.05324 &~4.98668 &~ 16 &~3.98668 &~4.92935 \\ \hline
  Erd\"os                   &~3.08843 &~ ~9 &~2.08843 &~3.06724 &~3.06705 &~ ~8 &~2.06705 &~3.05622 \\ \hline
  Routeview                 &~4.28302 &~ 12 &~3.28302 &~4.13371 &~4.80363 &~ 12 &~3.80363 &~4.66503 \\ \hline
  Homo  release 3.2.99      &~4.64504 &~ 12 &~3.64504 &~4.53609 &~3.96703 &~ 10 &~2.96703 &~3.94713 \\ \hline
  AS\_Caida\_20071105       &~4.24314 &~ 12 &~3.24314 &~4.11772 &~4.76795 &~ 12 &~3.76795 &~4.65617 \\ \hline
  Dimes 3/2010              &~3.43833 &~ ~9 &~2.43833 &~3.37664 &~3.35917 &~ ~8 &~2.35917 &~3.32159 \\ \hline
  Aqualab 12/2007- 09/2008  &~4.23183 &~ 12 &~3.23183 &~4.12775 &~4.54116 &~ 12 &~3.54116 &~4.4587  \\ \hline
  AS\_Caida\_20120601       &~4.10547 &~ 12 &~3.10547 &~4.0272  &~4.53051 &~ 12 &~3.53051 &~4.4896  \\ \hline
  itdk0304                  &~5.370078&~ 24 &~4.37008 &~5.3841  &~5.710122&~ 22 &~4.71012 &~5.82908 \\ \hline
  DBLB-coauth               &~5.57869 &~ 27 &~4.57869 &~5.53795 &~5.12724 &~ 22 &~4.12724 &~5.14932 \\ \hline
  Amazon                    &~8.81911 &~ 57 &~7.81911 &~8.78382 &~7.87004 &~ 42 &~6.87004 &~7.95201 \\ \hline
\end{tabular}\hspace*{-50mm}
\end{center}
\caption{Distortion results of non-contractive embedding of datasets into trees $H_{\ell}$ and $H'_{\ell}$. }
\label{tab:treeHl}
\vspace*{-0.5cm}\end{table}

\commentout{ 
\begin{table}
\footnotesize
\begin{center}
\hspace*{-50mm}
\begin{tabular}{ | c | p{1.3cm} | p{1.2cm} | p{1.3cm} | p{1.3cm}|}
    \hline
  Graph & avg. distortion & max distortion & avg. relative distortion & distance-weighted average distortion\\  \hline
  PPI &~5.29652 &~ 16&~ 4.29652&~ 5.2027\\ \hline
  Yeast &~3.79318 &~12 &~2.79318 &~ 3.74159\\ \hline
  DutchElite  &~ 6.53269&~ 20&~ 5.53269 &~ 6.4574\\ \hline
  EPA  &~4.06901 &~ 12&~ 3.06901 &~ 3.99447\\ \hline
  EVA  &~7.77752 &~ 18&~ 6.77752&~ 7.65544\\ \hline
  California &~4.98668 &~ 16&~3.98668&~ 4.92935\\ \hline
  Erd\"os  &~3.06705 &~ 8&~ 2.06705&~ 3.05622\\ \hline
  Routeview  &~4.80363 &~ 12&~ 3.80363 &~ 4.66503\\ \hline
  Homo  release 3.2.99 &~3.96703 &~ 10 &~ 2.96703&~3.94713 \\ \hline
  AS\_Caida\_20071105  &~ 4.76795&~ 12 &~ 3.76795 &~ 4.65617\\ \hline
  Dimes 3/2010  &~3.35917 &~ 8&~ 2.35917&~3.32159 \\ \hline
  Aqualab 12/2007- 09/2008  &~ 4.54116&~ 12 &~ 3.54116&~ 4.4587\\ \hline
  AS\_Caida\_20120601  &~ 4.54723&~ 39 &~ 3.54723&~4.49292 \\ \hline
  itdk0304 &~5.710122 &~ 22 &~4.71012 &~5.82908 \\ \hline
  DBLB-coauth &~5.12724 &~ 22 &~ 4.12724&~ 5.14932\\ \hline
  Amazon  &~7.87004 &~ 42 &~ 6.87004 &~ 7.95201 \\
  \hline
\end{tabular}\hspace*{-50mm}
\end{center}
\caption{Distortion results of embedding datasets into tree $H'_{\ell}$.}
\label{tab:treeHPl}
\vspace*{-0.5cm}\end{table}
}

\section{Tree-Breadth, Tree-Length and Tree-Stretch}\label{sec:tb}

There are 
two other graph parameters measuring metric tree likeness of a graph that
are based on the notion of tree-decomposition introduced by Robertson and Seymour in their work
on graph minors~\cite{RobSey86}.

A \emph{tree-decomposition} of a graph $G=(V,E)$ is a  pair
$(\{X_i|i\in I\},T=(I,F))$ where $\{X_i|i\in I\}$ is a collection of
subsets of $V$, called {\em bags}, and $T$ is a tree. The nodes of
$T$ are the bags $\{X_i|i\in I\}$ satisfying the following three
conditions (see Fig.~\ref{fig:graph-td}):
 \begin{enumerate}\vspace*{-2mm}
   \item $\bigcup_{i\in I}X_i=V$;
   \item for each edge $uv\in E$, there is a bag $X_i$ such that $u,v \in
   X_i$;
   \item for all $i,j,k \in I$, if $j$ is on the path from $i$ to $k$ in $T$, then $X_i \bigcap X_k\subseteq X_j$.
   Equivalently, this condition could be stated as follows: for all vertices $v\in V$, the set of bags $\{i\in I| v\in X_i\}$ induces a connected subtree $T_v$ of $T$.
 \end{enumerate}\vspace*{-2mm}
For simplicity we denote a tree-decomposition $\left(\{X_i|i\in
I\},T=(I,F)\right)$ of a graph $G$ by $\mathcal{T}(G)$.

\begin{figure}\vspace*{-.7cm}
        \centering
        \begin{subfigure}[b]{0.35\textwidth}
                \includegraphics[width=\textwidth]{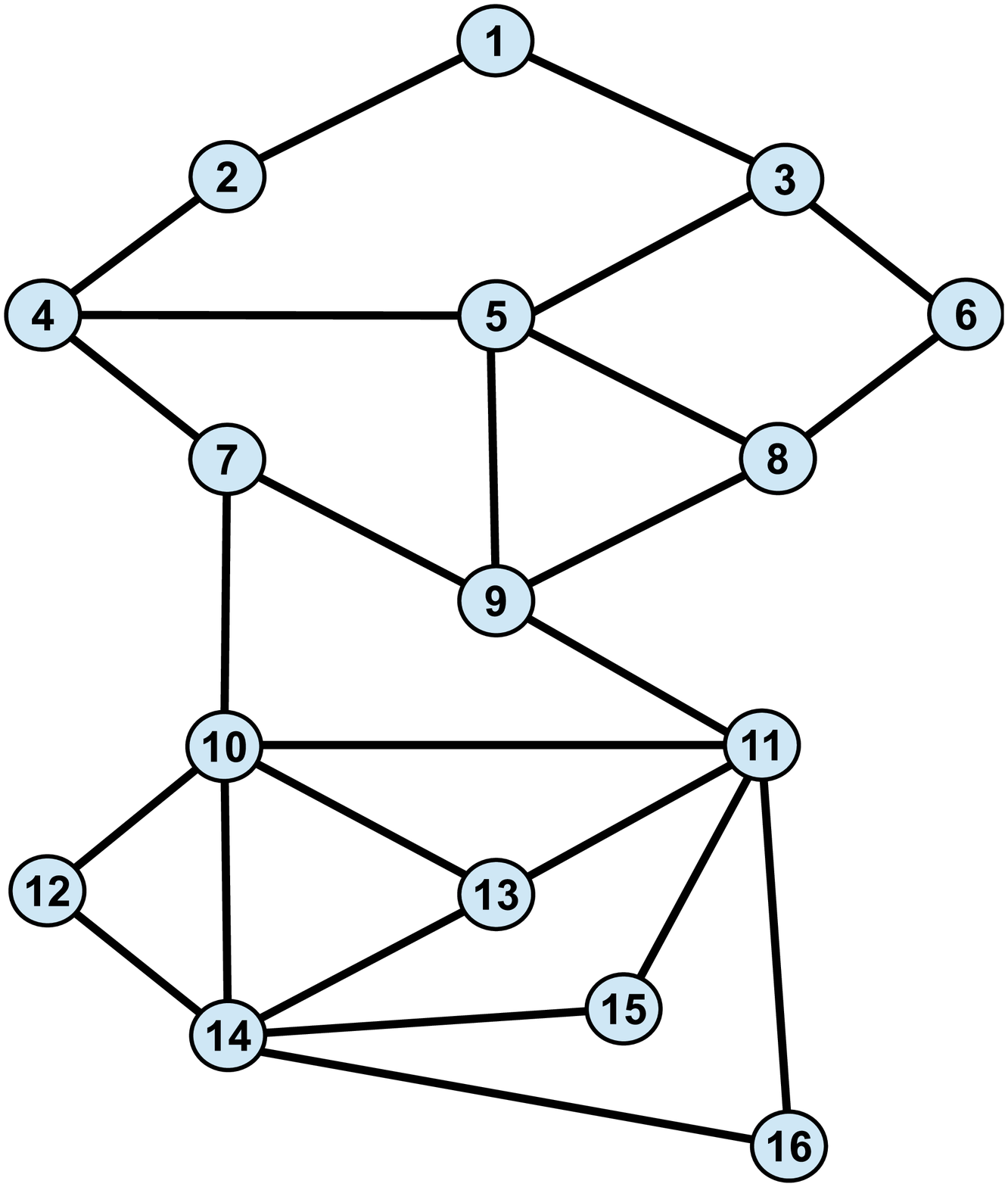}
                \vspace*{-.3cm}
                \caption{A graph $G$.}
                \label{fig:graphfortd}
        \end{subfigure}
          ~
        \begin{subfigure}[b]{0.35\textwidth}
                \includegraphics[width=\textwidth]{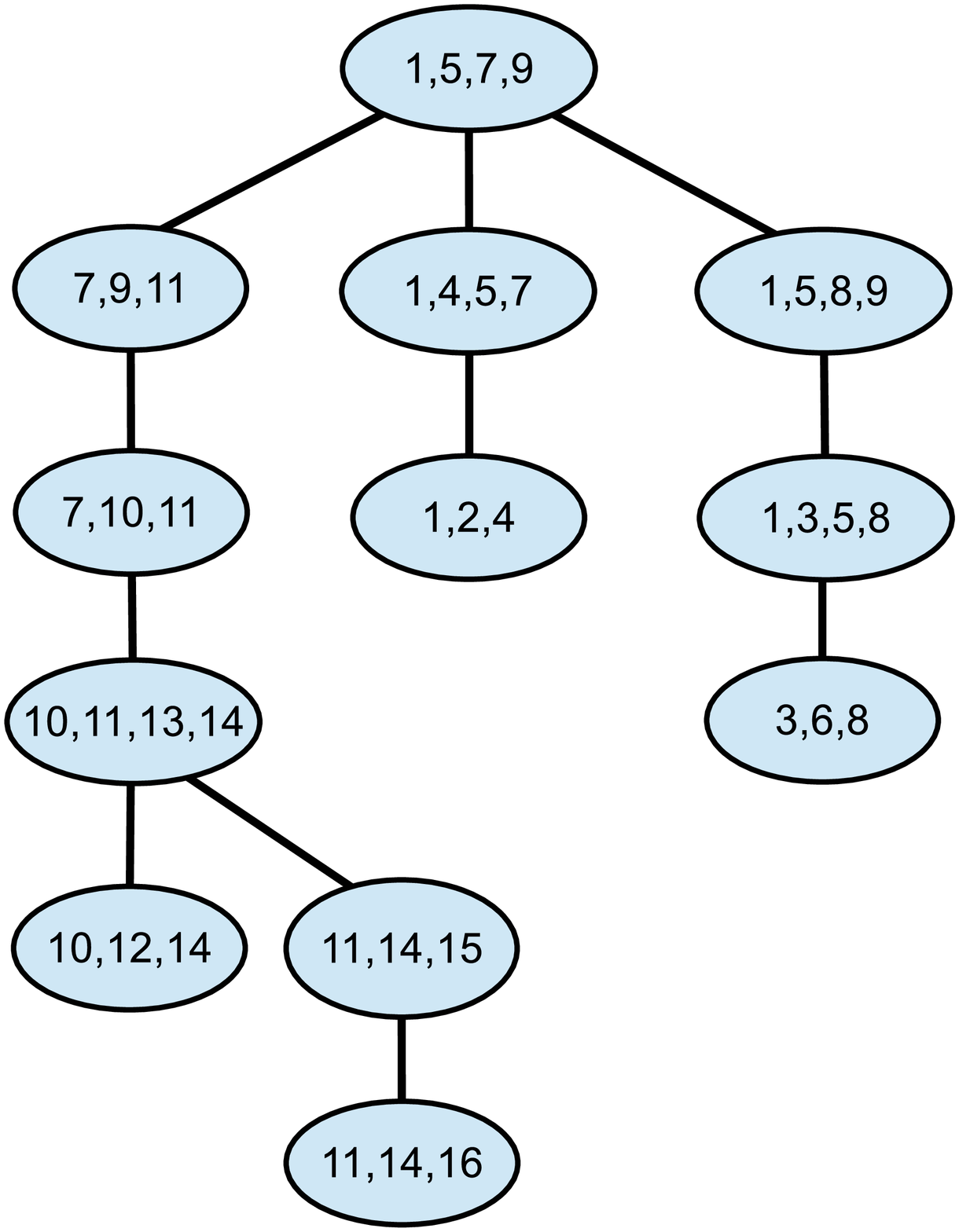}
                \vspace*{-.3cm}
                \caption{A tree-decomposition of $G$.}
                \label{fig:treedec}
        \end{subfigure}
         \caption{A graph and its tree-decomposition of width 3, of length 3, and of breadth 2.}
         \label{fig:graph-td}
\end{figure}

The \emph{width} of a tree-decomposition $\mathcal{T}(G)=(\{X_i|i\in
I\},T=(I,F))$ is $max_{i\in I}|X_i|-1$. The {\em tree-width} of a
graph $G$, denoted by $tw(G)$, is the minimum width over all
tree-decompositions $\mathcal{T}(G)$ of $G$~\cite{RobSey86}. The trees are
exactly the graphs with tree-width 1.

The {\em length} of a
tree-decomposition $\mathcal{T}(G)$ of a graph $G$ is $\lambda:=\max_{i\in
I}\max_{u,v\in X_i}d_G(u,v)$ (i.e., each bag $X_i$ has diameter at
most $\lambda$ in $G$). The {\em tree-length} of $G$, denoted by
$tl(G)$, is the minimum of the length  over all tree-decompositions
of $G$~\cite{DoGa2007}. The chordal graphs  are exactly the graphs
with tree-length 1. Note that these two graph parameters are not
related to each other. For instance, a clique
on $n$ vertices has tree-length 1 and tree-width $n-1$,
whereas a cycle on $3n$ vertices has tree-width 2 and tree-length
$n$. 
%
%
Analysis of few real-life networks (like Aqualab, AS\_Caida, Dimes) performed in \cite{conf/nca/MontgolfierSV11} shows that although those networks have small hyperbolicities, they all have sufficiently large tree-width due to well connected cores. As we demonstrate below, the tree-length of those graph datasets is relatively small.

The {\em breadth} of a tree-decomposition $\mathcal{T}(G)$ of a graph $G$ is the
minimum integer $r$ such that for every $i\in I$ there is a vertex
$v_i\in V$ with $X_i\subseteq B_r(v_i,G)$ (i.e., each bag $X_i$
can be covered by a disk $B_r(v_i,G):=\{u\in V(G) : d_G(u,v_i)\leq r
\}$ of radius at most $r$ in $G$). Note that vertex $v_i$ does not
need to belong to $X_i$. The {\em tree-breadth} of $G$, denoted by
$tb(G)$, is the minimum of the breadth over all
tree-decompositions of $G$ ~\cite{DBLP:conf/approx/DraganK11}. Evidently, for any graph $G$, $1\leq
tb(G)\leq tl(G)\leq 2 tb(G)$ holds. Hence, if one parameter is
bounded by a constant for a graph $G$ then the other parameter is
bounded for $G$ as well.

Clearly, in view of tree-decomposition  $\mathcal{T}(G)$ of $G$, the smaller parameters $tl(G)$ and $tb(G)$ of $G$ are, the closer graph $G$ is to a tree metrically.
Unfortunately, while graphs with tree-length 1 (as they are
exactly the chordal graphs) can be recognized in linear time, the
problem of determining whether a given graph has tree-length at most
$\lambda$ is NP-complete for every fixed $\lambda >1$ (see~\cite{Daniel10}). Judging from this result, it is conceivable that the
problem of determining whether a given graph has tree-breadth at most $\rho$ is NP-complete, too.


The following proposition from~\cite{DoGa2007} establishes a relationship between the tree-length and the cluster-diameter of a layering partition of a graph.
\begin{proposition}[\cite{DoGa2007}]\label{lem:tl-cluster-diam}
For every graph $G$ and any its vertex $s$,
$\Delta_s(G)/3 \leq tl(G) \leq \Delta_s(G)+1.$
\end{proposition}
Thus, the cluster-diameter $\Delta_s(G)$ of a layering partition provides easily computable bounds for the hard to compute parameter $tl(G)$.


One can prove similar inequalities relating the tree-breadth and the cluster-radius of a layering partition of a graph.
\begin{proposition}\label{lem:tb-cluster-radius} 
For every graph $G$ and any its vertex $s$,
$$\Delta_s(G)/6\leq R_s(G)/3 \leq tb(G) \leq R_s(G)+1\leq \Delta_s(G)+1.$$
Furthermore, a tree-decomposition of $G$ with breadth at most $3 tb(G)$ can be constructed in $O(n+m)$ time.
\end{proposition}


\emph{Proof.} The proof is 
similar to the proof  from~\cite{DoGa2007} of Proposition \ref{lem:tl-cluster-diam}.
First we show $R_s(G)/3 \leq tb(G)$. 
Let $\mathcal{T}(G)$ be  a tree-decomposition of $G$ with minimum breadth $tb(G)$.  Let $X_1X_2$ be an edge of $\mathcal{T}(G)$ and $\mathcal{T}_1,\mathcal{T}_2$ be subtrees of $\mathcal{T}(G)$ after removing the edge $X_1X_2$.
It is known~\cite{diestelGT} that set $I=X_1\bigcap X_2$ separates in $G$ vertices belonging to bags of $\mathcal{T}_1$ but not to $I$ from vertices belonging to bags of $\mathcal{T}_2$ but not to $I$.
Assume that $\mathcal{T}(G)$ is rooted at a bag containing vertex $s$, the source of layering partition $\mathcal{LP}(G,s)$. Let $C$ be a cluster from layer $L_i$
(i.e., $C=L_i^j$ for some $j=1,\cdots,p_i$). Let $Z$ be the nearest common ancestor of all  bags of $\mathcal{T}(G)$ containing vertices of $C$. Let $z$ be the vertex such that $Z\subseteq B_{tb(G)}(z,G)$.

%

%
Consider arbitrary vertex $x \in C$. Necessarily, there is a vertex $y \in C$ and two bags $X$ and $Y$ of $\mathcal{T}(G)$ containing vertices x and y, respectively, such that $Z=NCA_{\mathcal{T}(G)}(X,Y)$ (i.e., $Z$ is the nearest common ancestor of $X$ and $Y$ in $\mathcal{T}(G)$). Let $P$ be a shortest path of $G$ from $s$ to $x$. By the separator property above, $P$ intersects $Z$. See Fig.~\ref{fig:prop9} for an illustration. Let $a$ be a vertex of $P \bigcap Z$ closest to $s$ in $G$. Since both $x$ and $y$ belong to $C$, there exist a path $Q$ from $x$ to $y$ in $G$ using only intermediate vertices $w$ with $d_G(s,w) \geq i$. Let $b \in Q \cap Z$ (i.e. $Q$ intersects $Z$ at vertex $b$). We have $d_G(s,x)=i=d_G(s,a)+d_G(a,x)$ and $i \leq d_G(s,b) \leq d_G(s,a)+d_G(a,z)+d_G(z,b) \leq d_G(s,a)+2tb(G)$. Hence, $d_G(a,x) = i-d_G(s,a) \leq 2tb(G)$ and therefore $d_G(x,z) \leq d_G(x,a)+d_G(a,z) \leq 2tb(G) +tb(G)=3tb(G)$. Thus, any vertex $x$ of $C$ is at distance at most $3tb(G)$ from $z$ in $G$, implying $R_s(G)/3 \leq tb(G)$.

\begin{wrapfigure}[22]{l}{0.42\textwidth}
                \centering\vspace*{-.3cm}
                \includegraphics[width=0.35\textwidth]{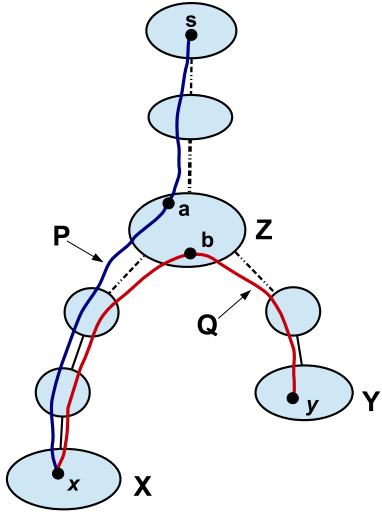}
         \vspace*{-.2cm}
         \caption{Illustration to the proof of Proposition \ref{lem:tb-cluster-radius}.}
         \label{fig:prop9}
\end{wrapfigure}
Note that, for the neighbor $x'$ of $x$ on $P$, $d(x',z)\leq 3 tb(G)-1$ must hold, i.e., $B_{3 tb(G)}(z,G)$ contains not only all vertices of $C=L_i^j$ but also all neighbors of vertices of $C$ laying in layer $L_{i-1}$. This fact will be useful in the second part of this proof.

Now we show that $tb(G) \leq R_s(G)+1$. Consider tree $\Gamma(G,s)$ of a layering partition $\mathcal{LP}(G,s)$ and assume $\Gamma(G,s)$ is rooted at node $\{s\}$. Let $p(C)$ be the parent of node $C$ in $\Gamma(G,s)$. Clearly, $\Gamma(G,s)$ satisfies already conditions 1 and 3 of tree-decompositions and only violates condition 2 as the edges joining vertices in different (neighboring) layers 
are not  yet covered by bags (which are the clusters in this case). We can obtain a tree-decomposition $\Gamma'$ from $\Gamma(G,s)$ as follows. $\Gamma'$ will have the same structure as $\Gamma(G,s)$, only the nodes of  $\Gamma(G,s)$  will slightly expand to cover additional edges of $G$ and form the bags of $\Gamma'$. To each node $C$ of $\Gamma(G,s)$ (assume $C\subseteq L_i$) 
we add all vertices from its parent $p(C)$ ($p(C)\subseteq L_{i-1}$) which are adjacent to vertices of $C$ in $G$. This expansion of $C$ results in a bag $C^+$ of $\Gamma'$ which, by construction, contains now also each edge $uv$ of $G$ with $u\in C\subseteq L_i$ and $v\in p(C)\subseteq L_{i-1}$.
Thus, $\Gamma'$ satisfies conditions 1 and 2 of  tree-decompositions. Also, if $C\subseteq B_{r}(z)$ for some vertex $z$ and integer $r$, then $C^+\subseteq B_{r+1}(z)$ must hold. Furthermore, each vertex $v$ of $G$ that was in a node $C$ now belongs to bag $C^+$ and to all bags formed from children of $C$ in $\Gamma(G,s)$ (and only to them). Hence,  all bags containing $v$ form a star in $\Gamma'$. All these indicate that $\Gamma'$ is a tree-decomposition of $G$ with breadth at most $R_s(G)+1$, i.e., $tb(G)\leq R_s(G)+1$.
%

Furthermore, as we indicated in the first part of this proof, for any cluster $C$ there is a vertex $z$ in $G$ such that  $C^+\subseteq B_{3 tb(G)}(z,G)$. The latter implies that the tree $\Gamma'$ obtained from $\Gamma(G,s)$ has breadth at most $3 tb(G)$.
Finally, since $\Gamma'$ is constructible in linear time and $R_s(G)\leq \Delta_s(G)\leq 2 R_s(G)$ holds for every graph $G$, the proposition follows. \qed

\medskip


Hence, the cluster-radius $R_s(G)$ of a layering partition provides easily computable bounds for the tree-breadth $tb(G)$ of a graph.
In Table ~\ref{tab:tb-bounds}, we show the corresponding lower and upper bounds on the tree-breadth for some of our datasets. The lower bound is obtained by dividing $R_s(G)$ by 3, the upper bound is obtained by calculating the breadth of the tree-decomposition $\Gamma'$.

\begin{table}
\footnotesize
\begin{center}
\begin{tabular}{ | c | c | c | c|}
    \hline
  Graph & $R_s(G)$ &lower bound & upper bound \\
  $G=(V,E)$ &  &on $tb(G)$  &on $tb(G)$  \\ \hline \hline
  PPI &4 & 2 & 5 \\ \hline
  Yeast &4 & 2&4 \\ \hline
  DutchElite  &6 & 2 & 6\\ \hline
  EPA  & 4& 2 & 4\\ \hline
  EVA  &5 &2 &5 \\ \hline
  California &4 & 2& 4\\ \hline
  Erd\"os  & 2 &1 &2 \\ \hline
  Routeview  &3  &1 &4\\ \hline
  Homo  release 3.2.99 & 3 &1 & 3\\ \hline
  AS\_Caida\_20071105 & 3 & 1 &3\\ \hline
  Dimes 3/2010  & 2 & 1 & 2 \\ \hline
  Aqualab 12/2007- 09/2008 & 3  & 1 & 3 \\ \hline
  AS\_Caida\_20120601 & 3 & 1 & 3 \\ \hline
  itdk0304 & 6 & 2 & 6 \\ \hline
  DBLB-coauth & 7 & 3 & 7 \\ \hline
  Amazon & 12 & 4 &  12 \\ \hline
\end{tabular}
\end{center}
\caption{Lower and upper bounds on the tree-breadth of our graph datasets.}
\label{tab:tb-bounds}
\vspace*{-0.5cm}\end{table}

\commentout{
\begin{table}
\footnotesize
\begin{center}
\begin{tabular}{ | c | c | c | c|}
    \hline
  Graph & $R_s(G)$ &lower bound & upper bound \\
  $G=(V,E)$ &  &on $tb(G)$  &on $tb(G)$  \\ \hline \hline
  PPI &4 & 2 & 5 \\ \hline
  Yeast &4 & 2&4 \\ \hline
  DutchElite  &6 & 2 & 6\\ \hline
  EPA  & 4& 2 & 4\\ \hline
  EVA  &5 &2 &5 \\ \hline
  California &4 & 2& 4\\ \hline
  Erd\"os  & 2 &1 &2 \\ \hline
  Routeview  &3  &1 &4\\ \hline
  Homo  release 3.2.99 & 3 &1 & 3\\ \hline
  AS\_Caida\_20071105 & 3 & 1 &3\\ \hline
  Dimes 3/2010  & 2 & 1 & 2 \\ \hline
  Aqualab 12/2007- 09/2008 & 3  & 1 & 3 \\ \hline
  AS\_Caida\_20120601 & 3 & 1 & 3 \\ \hline
  itdk0304 & 6 & 2 & 7 \\ \hline
  DBLB-coauth & 7 & 3 & 8 \\ \hline
  Amazon & 12 & 4 &  13 \\ \hline
\end{tabular}
\end{center}
\caption{Lower and upper bounds on the tree-breadth of our graph datasets.}
\label{tab:tb-bounds}
\vspace*{-0.5cm}\end{table}
}

Reformulating Proposition~\ref{lem:cluster-diam}, we obtain the following result.
\begin{proposition} \label{lem:treeH-tb}
For any graph $G=(V,E)$ and its canonic tree $H=(V,F)$ the following is true:
\[\forall u,v \in V, ~~d_H(u,v)-2 \leq d_G(u,v) \leq d_H(u,v)+ 3~tl(G) \leq d_H(u,v)+ 6~tb(G).\]
\end{proposition}

Graphs with small tree-length or small tree-breadth have many other nice properties.
Every $n$-vertex graph with tree-length $tl(G)=\lambda$ has an additive $2\lambda$-spanner with $O(\lambda n+n \log
n)$ edges  and an additive $4\lambda$-spanner with $O(\lambda n)$ edges, both
constructible in polynomial time~\cite{DoDrGaYa}. Every $n$-vertex graph $G$
with $tb(G)=\rho$ has a system of at most $\log_2 n$ collective additive tree
$(2\rho\log_2 n)$-spanners constructible in polynomial time~\cite{DBLP:conf/sofsem/DraganA13}.
Those graphs also enjoy a $6\lambda$-additive routing labeling scheme with $O(\lambda\log^2n)$ bit
labels and $O(\log\lambda)$ time routing protocol~\cite{DBLP:journals/jgaa/Dourisboure05}, and
a $(2\rho\log_2 n)$-additive routing labeling scheme with $O(\log^3n)$ bit
labels and $O(1)$ time routing protocol with $O(\log n)$ message initiation time (by combining results of~\cite{DBLP:conf/sofsem/DraganA13} and~\cite{DraganYC06}). See Section \ref{appl} for some details.

Here we elaborate a little bit more on a connection established in~\cite{DBLP:conf/approx/DraganK11} between the tree-breadth and the tree-stretch of a  graph (and the corresponding tree $t$-spanner problem).

The \emph{tree-stretch} $ts(G)$ of a graph $G=(V,E)$ is the smallest number $t$ such that $G$ admits a {\em spanning} tree $T=(V,E')$ with $d_T(u,v) \leq t  d_G(u,v)$ for every $u,v \in V.$ $T$ is called a {\em tree $t$-spanner} of $G$ and the problem of finding such tree $T$ for $G$ is known as the {\em tree $t$-spanner problem}. Note that as $T$ is a spanning tree of $G$, necessarily $d_G(u,v) \leq d_T(u,v)$ and $E' \subseteq E$. The latter makes the tree-stretch parameter different from the tree-distortion  where new (not from graph) edges can be used to build a tree. It is known that the tree $t$-spanner problem is NP-hard~\cite{CaiC95}. The best known approximation algorithms have approximation ratio of $O(\log n)$~\cite{EmekP08,DBLP:conf/approx/DraganK11}.

The following two results were obtained in~\cite{DBLP:conf/approx/DraganK11}.

\begin{proposition}[\cite{DBLP:conf/approx/DraganK11}] \label{prop:ts-tb}
For every graph $G$, $tb(G) \leq \lceil{ts(G)/2}\rceil$ and $tl(G) \leq ts(G)$.
\end{proposition}

\begin{proposition}[\cite{DBLP:conf/approx/DraganK11}]\label{prop:ts-log-appr}
For every $n$-vertex graph $G$, $ts(G) \leq 2tb(G) \log_2n$. Furthermore, a spanning tree $T$ of $G$ with $d_T(u,v) \leq 2 tb(G) \log_2n ~d_G(u,v)$, for every $u,v \in V,$ can be constructed in polynomial time.
\end{proposition}

Proposition \ref{prop:ts-log-appr} is obtained by showing that every $n$-vertex graph $G$ with $tb(G)=\rho$ admits a tree $(2\rho\log_2 n)$-spanner constructible in polynomial time. Together with Proposition \ref{prop:ts-tb}, this provides a $\log_2 n$-approximate solution for the tree $t$-spanner problem in general unweighted graphs.

We conclude this section with two other inequalities establishing relations between the tree-stretch and the tree-distortion and hyperbolicity of a graph.

\begin{proposition}[\cite{DraganWG2013}] \label{prop:ts-td}
For every graph $G$,  $tl(G) \leq td(G)\leq ts(G) \leq 2 td(G) \log_2n$.
\end{proposition}

\begin{proposition}[\cite{DraganWG2013}] \label{prop:ts-hyperb}
For every $\delta$-hyperbolic graph $G$, $ts(G) \leq O(\delta\log^2n)$.
\end{proposition}

Proposition \ref{prop:ts-td} says that if a graph $G$ is non-contractively embeddable into a tree with distortion $td(G)$ then it is embeddable into a spanning tree with stretch at most $2 td(G) \log_2n$. Furthermore, a spanning tree with stretch at most $2 td(G) \log_2n$ can be constructed in polynomial time.
Proposition \ref{prop:ts-hyperb} says that every $\delta$-hyperbolic graph $G$ admits a tree $O(\delta\log^2n)$-spanner. Furthermore, such a spanning tree for a  $\delta$-hyperbolic graph can be constructed in polynomial time.

\section{Use of Metric Tree-Likeness} \label{appl} 
As we have mentioned earlier, metric tree-likeness of a graph is useful in a number of ways. Among other advantages, it allows to design compact and efficient approximate distance labeling and routing labeling schemes, 
fast and accurate estimation of the diameter and the radius of a graph. In this section, we elaborate more on these applications.
In general, low distortion embedability of a graph $G$ into a tree $T$ allows to solve approximately many distance related problems on $G$ by first solving them on the tree $T$ and then interpreting that solution on $G$. 
\vspace*{-2mm}

\subsection{Approximate distance queries}\vspace*{-2mm}
%
Commonly, when one makes a query concerning a pair of vertices in a
graph (adjacency, distance, shortest route, etc.), one needs to make a
global access to the structure storing that information. A compromise to this approach is to
store enough information locally in a label associated with a vertex
such that the query can be answered using only the information in
the labels of two vertices in question and nothing else. Motivation
of localized data structure in distributed computing is surveyed and
widely discussed in~\cite{Pel,GaPe}.

Here, we are mainly interested in the distance and routing labeling
schemes,  introduced by Peleg (see, e.g.,
\cite{Pel}). 
{\it Distance labeling schemes} 
are schemes that
label the vertices of a graph with short labels in such a way that
the distance between any two vertices $u$ and $v$ can be determined
or estimated efficiently by merely inspecting the labels of $u$ and $v$, without using any
other information. {\it Routing labeling schemes} 
are schemes that label the vertices of a graph with short labels in
such a way that given the label of a source vertex and the label of
a destination, it is possible to compute efficiently  the port
number of the edge from the source that heads in the direction of
the destination.
%
%

It is known that $n$-vertex trees enjoy a distance labeling  scheme where each vertex is assigned a $O(\log^2 n)$-bit label such that given labels of two vertices the distance between them can be inferred in constant time~\cite{DBLP:conf/wg/Peleg99}.
We can use for our datasets their canonic trees to compactly and distributively encode their approximate distance information. Given a graph dataset $G$, we first compute in linear time its canonic tree $H$. Then, we preprocess $H$ in $O(n\log n)$ time (see~\cite{DBLP:conf/wg/Peleg99}) to assign each vertex $v\in V$ an $O(\log^2 n)$-bit distance label. Given two vertices $u,v\in V$, we can compute in $O(1)$ time the distance $d_H(u,v)$ from their labels and output this distance as a good  estimate for the distance between $u$ and $v$ in $G$.

\begin{figure}[htb]
        \centering
        \footnotesize
        \vspace*{-2mm}
         \subcaptionbox{Percentage of vertex pairs whose distance was distorted only up-to a given value.\label{tab:treeHDist}}
            [0.55\textwidth]{
                \begin{tabular}{ | c | c | c | c | c |c | c| }
                    \hline
                     \multicolumn{1}{|c|}{Graph}   &    \multicolumn{6}{c|}{distortion} \\
                      $G=V,E)$ & \multicolumn{1}{c}{= 1} & \multicolumn{1}{c}{$<$ 1.2} &\multicolumn{1}{c}{$<$ 1.3} & \multicolumn{1}{c}{$<$ 1.5} & \multicolumn{1}{c}{$<$ 2}& \multicolumn{1}{c|}{$<$ 2.2} \\  \hline  \hline
                    PPI & 20.41 & 37.68 & 47.90 & 65.93 & 90.68 & 96.37 \\ \hline
                    Yeast &  31.51 & 38.45 & 53.22 & 72.30 & 91.03 & 98.55 \\ \hline
                    DutchElite  & 23.13  & 27.99 & 42.97 &64.60 & 88.71 &  95.44\\ \hline
                    EPA  & 44.87 & 50.83 & 65.50 & 76.52 & 91.82 & 98.68 \\ \hline
                    EVA  & 52.92 & 73.37 & 82.68 & 92.83 & 99.12 & 99.88\\ \hline
                    California & 30.25 & 40.21 & 51.89 & 64.53 & 88.97 & 98.06 \\ \hline
                    Erd\"os  &  88.34 & 88.34 & 89.84 & 96.99 & 99.55 & 99.98 \\ \hline
                    Routeview  & 42.28 & 44.75 & 58.17 & 81.94 & 96.40 & 99.85 \\ \hline
                    Homo  release 3.2.99 & 72.01 & 72.13 & 73.48 & 79.08 & 90.79 & 99.97 \\ \hline
                    AS\_Caida\_20071105  & 43.15 & 46.60 & 62.39 & 84.54 & 95.68 & 99.90 \\ \hline
                    Dimes 3/2010  & 49.84  & 50.06 & 56.77 & 89.30 & 97.05 & 99.99 \\ \hline
                    Aqualab 12/2007- 09/2008  &  32.54 & 33.23 &  44.61& 76.46 & 95.93 & 99.98 \\ \hline
                    AS\_Caida\_20120601  & 57.15  & 59.57 & 71.82 & 89.58 & 98.65 & 99.98 \\ \hline
                    itdk0304 &  4.60 & 15.18 &23.67  & 42.54 & 81.98 & 93.55 \\ \hline
                    DBLB-coauth& 3.59  &  12.08 &17.60  & 30.64 & 67.92 & 83.10 \\ \hline
                    Amazon & 0.63 & 2.67 & 4.57 & 10.16 & 33.10 & 46.53\\
                    \hline
                    \end{tabular}
                }~
                \subcaptionbox{Accumulative frequency chart.\label{fig:distortion-chart}}
             [0.44\textwidth]
            {\includegraphics[width=0.44\textwidth, height=7.22cm]{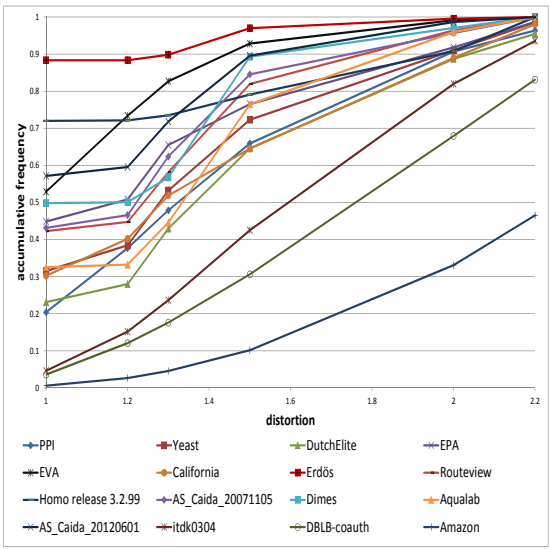}}
\caption{Distortion distribution for embedding of a graph  dataset into its canonic tree $H$.}\label{fig:distortion-distirbution}
\end{figure}

On Fig. \ref{fig:distortion-distirbution}, we  demonstrate how accurate canonic trees represent pairwise distances in our datasets. For a given number $\epsilon\geq 1$, we show how many vertex pairs had a distortion less than $\epsilon$, i.e., pairs $u,v\in V$ with $\max\{\frac{d_H(u,v)}{d_G(u,v)}, \frac{d_G(u,v)}{d_H(u,v)}\}<\epsilon.$
We can see that $H$ approximates distances for most vertex pairs with a high level of accuracy. 
Exact graph distances were preserved in $H$ for at least 40\% of pairs in 8 datasets (EPA, EVA, Erd\"os, Routeview, Homo, AS\_Caida\_20071105, Dimes 3/2010 and AS\_Caida\_20120601). At least 50\% of pairs of 6 datasets have distance distortion in $H$ less than 1.2. At least 60\% of pairs for 6 datasets have distance distortion less than 1.3. At least 70\% of pairs of 10 datasets have distance distortion less than 1.5. At least 80\% of pairs of 14 datasets have distance distortion less than 2. At least 90\% of pairs of 14 datasets have distance distortion less than 2.2. For the DBLB-coauth dataset, 80\% (90\%) of pairs embed into $H$ with distortion no more than 2.2 (2.4, respectively; not shown on table). For the Amazon dataset, 80\% (90\%) of pairs embed into $H$ with distortion no more than 3.2 (3.8, respectively; not shown on table).

Hence, using embeddings of our datasets into their canonic trees, we obtain a compact and efficient approximate distance labeling scheme for them. Each vertex of a graph dataset $G$ gets $O(\log^2 n)$-bit label from the canonic tree and the distance between any two vertices of $G$ can be computed with a good level of accuracy in constant time from their labels only.
\vspace*{-2mm}

\subsection{Approximating optimal routes}\vspace*{-2mm}
First we formally define approximate routing labeling schemes. A family $\Re$
of graphs is said to have an {\it $l(n)$ bit $(s,r)$-approximate
routing labeling scheme} if there exist a function $L,$ labeling the
vertices of each $n$-vertex graph in $\Re$  with distinct labels of
up to $l(n)$ bits, and  an efficient algorithm/function $f,$ called
the {\em routing decision} or {\it routing protocol}, that given the
label of a current vertex $v$ and the label of the destination
vertex (the header of the packet), decides in time polynomial in the
length of the given labels and using only those two labels, whether
this packet has already reached its destination, and if not, to
which neighbor of $v$ to forward the packet. Furthermore, the
routing path from any source $s$ to any destination $t$ produced  by
this scheme in a graph $G$ from $\Re$ must have the length at most
$s\cdot d_G(s,t)+r$. For simplicity,  $(1,r)$-approximate labeling
schemes (distance or routing) are called {\it $r$-additive labeling
schemes}, and $(s,0)$-approximate labeling schemes are called {\it
$s$-multiplicative labeling schemes}.

A very good routing labeling  scheme exists for trees~\cite{DBLP:conf/spaa/ThorupZ01}. An $n$-vertex tree can be preprocessed in $O(n\log n)$ time so that each vertex is assigned an $O(\log n)$-bit routing label. Given the label of a source vertex and the label of a destination, it is possible to compute in constant time   the port number of the edge from the source that lays on the (shortest) path to the destination.

Unfortunately, a canonic tree $H$ of a graph $G$ is not suitable for approximately routing in $G$; $H$ may have artificial edges (not coming from $G$) and therefore a path of $H$ from a source to a destination may not be available for routing in $G$. To reduce the problem of routing in $G$ to routing in a tree $T$, tree $T$ needs to be a spanning tree of $G$. Hence, a spanning tree $T$ of $G$ with minimum stretch (i.e., a tree $t$-spanner of $G$ with $t=ts(G)$) would be a perfect choice. Unfortunately, finding a tree $t$-spanner of a graph with minimum $t$ is an NP-hard problem.

For our graph datasets, one can exploit the facts that they have small tree-breadth/tree-length and/or small hyperbolicity.

If the tree-breadth of an $n$-vertex  graph $G$ is $\rho$ then, by a result from~\cite{DBLP:conf/approx/DraganK11}, $G$ admits a tree $(2\rho\log_2 n)$-spanner constructible in polynomial time. Hence, $G$ enjoys a $(2\rho\log_2 n)$-multiplicative routing labeling scheme with $O(\log n)$ bit labels and $O(1)$ time routing protocol (routing is essentially done  in that tree spanner). Another result for graphs with $tb(G)=\rho$, useful for designing routing labeling schemes, is presented in~\cite{DBLP:conf/sofsem/DraganA13}. It states that every $n$-vertex graph $G$ with $tb(G)=\rho$ has a system of at most $\log_2 n$ collective additive tree
$(2\rho\log_2 n)$-spanners, i.e., a system $\cal{T}$ of at most $\log_2 n$ spanning trees of $G$ such that for any two vertices $u,v$ of $G$ there is a tree $T$ in $\cal{T}$ with $d_T(u,v)\leq d_G(u,v)+ 2\rho\log_2 n$. Furthermore, such a system $\cal{T}$ for $G$ can be constructed in polynomial time~\cite{DBLP:conf/sofsem/DraganA13}. By combining this with a result from~\cite{DraganYC06}, we obtain that every $n$-vertex graph $G$ with $tb(G)=\rho$
enjoys a $(2\rho\log_2 n)$-additive routing labeling scheme with $O(\log^3n)$ bit
labels and $O(1)$ time routing protocol with $O(\log n)$ message initiation time. The approach of~\cite{DraganYC06} is to assign to each vertex of $G$ a label with $O(\log^3n)$ bits (distance and routing labels coming from $\log_2 n$ spanning trees) and then, using the label of source vertex $v$ and the label of destination vertex $u$, identify in $O(\log n)$ time the best spanning tree in $\cal{T}$ to route from $v$ to $u$.

If the tree-length of an $n$-vertex  graph $G$ is $\lambda$ then, by result from~\cite{DBLP:journals/jgaa/Dourisboure05}, $G$ enjoys a $6\lambda$-additive routing labeling scheme with $O(\lambda\log^2n)$ bit labels and $O(\log\lambda)$ time routing protocol.


If the hyperbolicity of an $n$-vertex  graph $G$ is $\delta$ then, by result from~\cite{ChDrEsRout}, $G$ enjoys an $O(\delta\log
n)$-additive routing labeling scheme with $O(\delta\log^2n)$ bit labels and $O(\log\delta)$ time routing protocol.
Note that for any graph $G$, the hyperbolicity of $G$ is at most its tree-length \cite{DBLP:conf/compgeom/ChepoiDEHV08}.


Thus, for our graph datasets, there exists a very compact labeling scheme (at most $O(\log^2 n)$ or $O(\log^3 n)$ bits per vertex) that encodes logarithmic length routes between any pair of vertices, i.e., routes of length at most $d_G(u,v)+\min\{O(\delta\log n), 6\lambda, 2\rho\log_2 n\}\leq diam(G)+O(\log n)\leq O(\log n)$ for each vertex pair $u,v$ of $G$. The latter implies very good navigability of our graph datasets. Recall that, for our graph datasets, $diam(G)\leq O(\log n)$ holds.

\vspace*{-2mm}

\begin{table}\vspace*{-2mm}
\footnotesize
\begin{center}
\begin{tabular}{ | c | c | c | c | c|}
    \hline
  Graph     & diameter  &  radius    & $\#$ of  BFS scans & estimated radius     \\
  $G=(V,E)$ & $diam(G)$ &  $rad(G)$  & needed to get      & or $ecc(\cdot)$ of a \\
            &           &            &   $diam(G)$        & middle vertex        \\ \hline\hline
  PPI& 19 & 11 & 3 & 12\\ \hline
  Yeast & 11 & 6 & 3 & 6 \\ \hline
  DutchElite & 22  & 12 & 4 & 13\\ \hline
  EPA & 10 & 6 & 2 & 7\\ \hline
  EVA &18 & 10 & 2 & 10\\ \hline
  California & 13  & 7 & 2 & 8 \\ \hline
  Erd\"os & 4  & 2 & 2 & 3 \\ \hline
  Routeview & 10 & 5 & 2 & 5\\ \hline
  Homo  release 3.2.99 & 10 & 5 & 2 & 6 \\ \hline
  AS\_Caida\_20071105 & 17 & 9 & 2 & 9\\ \hline
  Dimes 3/2010& 8  & 4 & 2 & 5\\ \hline
  Aqualab 12/2007- 09/2008 & 9  & 5 & 2 & 5 \\ \hline
  AS\_Caida\_20120601 & 10 & 5 & 2 & 5 \\ \hline
  itdk0304& 26  & 14 & 2 & 15 \\ \hline
  DBLB-coauth & 23 & 12 & 2 & 14\\ \hline
  Amazon & 47  & 24 & 2 & 26 \\
  \hline
\end{tabular}
\end{center}
\caption{Estimation of diameters and radii. }
\label{tab:diamRadius}
\vspace*{-0.5cm}\end{table}

\subsection{Approximating diameter and radius}\vspace*{-2mm}
Recall that the {\em eccentricity} of a vertex $v$ of a graph $G$, denoted by $ecc(v)$, is the maximum distance from $v$ to any other vertex of $G$, i.e., $ecc(v):=\max_{u\in V} d_G(v,u)$. The {\em diameter} $diam(G)$ of $G$ is the largest eccentricity of a vertex in $G$, i.e., $diam(G):=\max_{v\in V}ecc(v)= \max_{v,u\in V}d_G(u,v)$.  The {\em radius} $rad(G)$ of $G$ is the smallest eccentricity of a vertex in $G$, i.e., $rad(G):=\min_{v\in V}ecc(v)$. A vertex $c$ of $G$ with $ecc(v)=rad(G)$ (i.e., a smallest eccentricity vertex) is called a {\em central vertex} of $G$. The {\em center} $C(G)$ of $G$ is the set of all central vertices of $G$. Let also $F(v):=\{u\in V: d_G(v,u)=ecc(v)\}$ be the set of vertices of $G$ furthest from $v$.

In general (even unweighted) graphs, it is still an open problem  whether the diameter and/or the radius of a graph $G$ can be computed faster than the time needed to compute the entire distance matrix of $G$ (which requires $O(n m)$ time for a general unweighted graph).
On the other hand, it is known that both, the diameter and the radius, of a tree $T$ can be calculated in linear time. That can be done by using 2 Breadth-First-Search (BFS) scans as follows. Pick an arbitrary vertex $u$ of $T$. Run a
BFS starting from $u$ to find $v\in F(u).$ Run a second BFS starting from $v$ to find $w\in F(v).$
Then $d_T(v,w)=diam(T),$ i.e., $v,w$ is a {\em diametral pair} of $T$, and $rad(G)=\lfloor{(d_T(v,w)+1)/2}\rfloor$. To find
the center of $T$ it suffices to take one or two
adjacent middle vertices of the $(v,w)$-path of $T$.

Interestingly, in~\cite{DBLP:conf/compgeom/ChepoiDEHV08}, Chepoi et al.  established that this approach of 2 BFS-scans can be adapted to
provide fast (in linear time) and accurate approximations of the diameter, radius,
and center of any finite set $S$ of $\delta$-hyperbolic geodesic spaces and graphs. In particular, for a $\delta$-hyperbolic graph $G$, it was shown that if $v\in F(u)$ and $w\in F(v),$ then
$d_G(v,w)\ge diam(G)-2\delta$ and $rad(G)\le \lfloor{(d_G(v,w)+1)/2}\rfloor+3\delta.$ Furthermore, the center $C(G)$ of $G$ is contained in the ball
of radius $5\delta+1$ centered at a middle vertex $c$ of any shortest path connecting $v$ and $w$ in $G$. 


Since our graph datasets have small hyperbolicities, according to~\cite{DBLP:conf/compgeom/ChepoiDEHV08}, few (2, 3, 4, ...) BFS-scans, each next starting at a vertex last visited by the previous scan) should provide a pair of vertices $x$ and $y$ such that $d_G(x,y)$ is close to the diameter $diam(G)$ of $G$. Surprisingly (see Table~\ref{tab:diamRadius}), few BFS-scans were sufficient to get exact diameters of all of our datasets: for 13 datasets, 2 BFS-scans (just like for trees) were sufficient to find the exact diameter of a graph. Two datasets needed 3 BFS-scans to find the diameter, and only one dataset required 4 BFS-scans to get the diameter. We also computed the eccentricity of a middle vertex of a longest shortest path produced by these few BFS-scans and reported this eccentricity as an estimation for the graph radius.
It turned out that the eccentricity of that middle vertex was equal to the exact radius for 6 datasets, was only one apart from the exact radius for 8 datasets, and only for 2 datasets was two units apart from the exact radius.


\section{Conclusion} \label{sec:concl} 
Based on solid theoretical foundations, we presented strong evidences that a number of real-life networks, taken from different domains like Internet measurements, biological datasets, web graphs, social and collabora\-tion networks, exhibit metric tree-like structures.
We investigated a few graph parameters, namely, the tree-distortion and the tree-stretch, the tree-length and the tree-breadth, the Gromov's hyperbolicity, the cluster-diameter and the cluster-radius in a layering partition of a graph, which capture and quantify this phenomenon of
being metrically close to a tree. Recent advances in theory allowed us to calculate or accurately estimate these parameters for sufficiently large networks. All these parameters are at most constant or (poly)logarithmic factors apart from each other. Specifically, graph parameters $td(G)$, $tl(G)$, $tb(G)$, $\Delta_s(G)$, $R_s(G)$ are within small constant factors from each other. Parameters $ts(G)$ and $\delta(G)$ are within factor of at most $O(\log n)$ from $td(G)$, $tl(G)$, $tb(G)$, $\Delta_s(G)$, $R_s(G)$. Tree-stretch $ts(G)$ is within factor of at most $O(\log^2 n)$ from hyperbolicity $\delta(G)$. One can summarize those relationships with 
the following chains of inequalities: 
$$\delta(G)\leq \Delta_s(G)\leq O(\delta(G)\log n);~~R_s(G)\leq \Delta_s(G)\leq 2R_s(G);~~tb(G)\leq tl(G)\leq 2tb(G);$$
$$\delta(G)\leq tl(G)\leq td(G)\leq ts(G)\leq 2tb(G)\log_2 n\leq O(\delta(G)\log^2 n);$$
$$tl(G)-1\leq \Delta_s(G)\leq 3tl(G)\leq 3td(G)\leq 3(2\Delta_s(G)+2);$$
$$tb(G)-1\leq R_s(G)\leq 3tb(G)\leq 3\lceil{ts(G)/2}\rceil.$$
If one of these parameters or its average version has 
small value for a large scale network, we say that that network has a metric tree-like structure. Among these parameters theoretically smallest ones are $\delta(G)$, $R_s(G)$  and $tb(G)$ ($tb(G)$ being at most $R_s(G)+1$). Our experiments showed that average versions of $\Delta_s(G)$ and of $td(G)$ have  also very small values for the investigated graph datasets.

In Table \ref{tab:allMeasures}, we provide a summary of metric 
tree-likeness measurements calculated for our datasets. Fig.~\ref{fig:tree-likeness-1charts}  shows four important metric tree-likeness measurements (scaled) in comparison. Fig.~\ref{fig:tree-likeness-charts} gives pairwise dependencies between those measurements (one as a function of another).


\begin{table}
\footnotesize
\begin{center}
\begin{tabular}{ | c | c | c | c| c | c | c | c | c | c|}  
    \hline
  Graph     & diameter  &  radius       & cluster-         & average    & $\delta(G)$ & Tree $H$   & $H_{\ell}$ & $H'_{\ell}$ &cluster-\\
  $G=(V,E)$ & $diam(G)$ & $rad(G)$      & diameter         &   diameter    &           & average    & average  & average & radius     \\
            &           &               & $\Delta_s(G)$    & of clusters in  &           & distortion\textsuperscript{*}
            & distortion & distortion & $R_s(G)$\\
        & & & & $\mathcal{LP}(G,s)$  &   & (round.)  &  &  &  \\      \hline \hline
  PPI                       & 19 & 11 & 8  & ~~0.118977384   & 3.5 & 1.38471  & 5.70566 & 5.29652  & 4\\ \hline
  Yeast                     & 11 & 6  & 6  & ~~0.119575699   & 2.5 & 1.32182  & 4.37781 & 3.79318  & 4\\ \hline
  DutchElite                & 22 & 12 & 10 & ~~0.070211316   & 4   & 1.41056  & 5.45299 & 6.53269  & 6\\ \hline
  EPA                       & 10 & 6  & 6  & ~~0.06698375    & 2.5 & 1.26507  & 4.50619 & 4.06901  & 4\\ \hline
  EVA                       & 18 & 10 & 9  & ~~0.031879981   & 1   & 1.13766  & 5.83084 & 7.77752  & 5\\ \hline
  California                & 13 & 7  & 8  & ~~0.092208234   & 3   & 1.35380  & 4.15785 & 4.98668  & 4\\ \hline
  Erd\"os                   & 4  & 2  & 4  & ~~0.001113232   & 2   & 1.04630  & 3.08843 & 3.06705  & 2\\ \hline
  Routeview                 & 10 & 5  & 6  & ~~0.063264697   & 2.5 & 1.23716  & 4.28302 & 4.80363  & 3\\ \hline
  Homo  release 3.2.99      & 10 & 5  & 5  & ~~0.03432595    & 2   & 1.18574  & 4.64504 & 3.96703  & 3\\ \hline
  AS\_Caida\_20071105       & 17 & 9  & 6  & ~~0.056424679   & 2.5 & 1.22959  & 4.24314 & 4.76795  & 3\\ \hline
  Dimes 3/2010              & 8  & 4  & 4  & ~~0.056582633   & 2   & 1.19626  & 3.43833 & 3.35917  & 2\\ \hline
  Aqualab 12/2007- 09/2008  & 9  & 5  & 6  & ~~0.05826733    & 2   & 1.28390  & 4.23183 & 4.54116  & 3\\ \hline
  AS\_Caida\_20120601       & 10 & 5  & 6  & ~~0.055568105   & 2   & 1.16005  & 4.10547 & 4.53051  & 3\\ \hline
  itdk0304                  & 26 & 14 & 11 & ~~0.270377048   & --  & 1.57126  & 5.370078& 5.710122 & 6\\ \hline
  DBLB-coauth               & 23 & 12 & 11 & ~~0.45350002    & --  & 1.74327  & 5.57869 & 5.12724  & 7\\ \hline
  Amazon                    & 47 & 24 & 21 & ~~0.489056144   & --  & 2.47109  & 8.81911 & 7.87004  & 12\\
  \hline
  \multicolumn{10}{l}{\textsuperscript{*}\footnotesize{
  $=\frac{\mbox{avg. distortion right}\times \# \mbox{right pairs }+\mbox{ avg. distortion left}\times\# \mbox{left pairs }+ \# \mbox{undistorted pairs}}{\binom{n}{2}}$}}
\end{tabular}
\end{center}
\caption{Summary of tree-likeness measurements.}
\label{tab:allMeasures}
\vspace*{-0.5cm}\end{table}


From the experiment results we observe that in almost all cases the
measurements seem to be monotonic with respect to each others. The smaller one measurement is for a given dataset, the smaller the other measurements are. There are also a few exceptions. For example, EVA dataset has relatively large cluster-diameter, $\Delta_s(G)=9$, but small hyperbolicity, $\delta(G)=1$. On the other hand, Erd\"os dataset has $\Delta_s(G)=4$ while its hyperbolicity $\delta(G)$ is equal to 2 
(see Figure \ref{fig:delta-clusDiam}). Yet Erd\"os dataset has better embedability (smaller average distortions) to trees $H, H_{\ell}$ and $H'_{\ell}$ than that of EVA, suggesting that the (average) cluster-diameter may have greater impact on the embedability into trees
$H, H_{\ell}$ and $H'_{\ell}$.
%
Comparing the measurements of Erd\"os vs. Homo  release 3.2.99, we observe that both have the same hyperbolicity 2, but Erd\"os has
better embedability (average distortion) to trees $H, H_{\ell},H'_{\ell}$. This could be explained by smaller $\Delta_s(G)$ and average diameter of clusters in  Erd\"os dataset.
Comparing measurements of PPI vs. California (the same holds for AS\_Caida\_20071105 vs. AS\_Caida\_20120601), both have same $\Delta_s(G)$ and $R_s(G)$ values but California (AS\_Caida\_20120601) has smaller hyperbolicity and average diameter of clusters. We also observe that the datasets Routeview and AS\_Caida\_20071105 have same values of $\Delta_s(G)$, $R_s(G)$ and $\delta(G)$ but AS\_Caida\_20071105 has a relatively smaller average diameter of clusters. This could explain why AS\_Caida\_20071105 has relatively better embedability to $H, H_{\ell}$ and $H'_{\ell}$ than Routeview.
We can see that the difference in average diameters of clusters was relatively small, 
resulting 
in small difference in embedability.

From these observations, one can suggest that for classification of our datasets all these tree-likeness measurements are important, they collectively capture and explain metric tree-likeness of them. We suggest that metric tree-likeness measurements in conjunction with other local characteristics
of networks, such as the degree distribution and clustering coefficients, provide a more complete unifying picture of networks.  

\begin{figure}[htb]
\begin{center} 
 \begin{minipage}[b]{13cm}
        \centering
                \includegraphics[width=\textwidth]{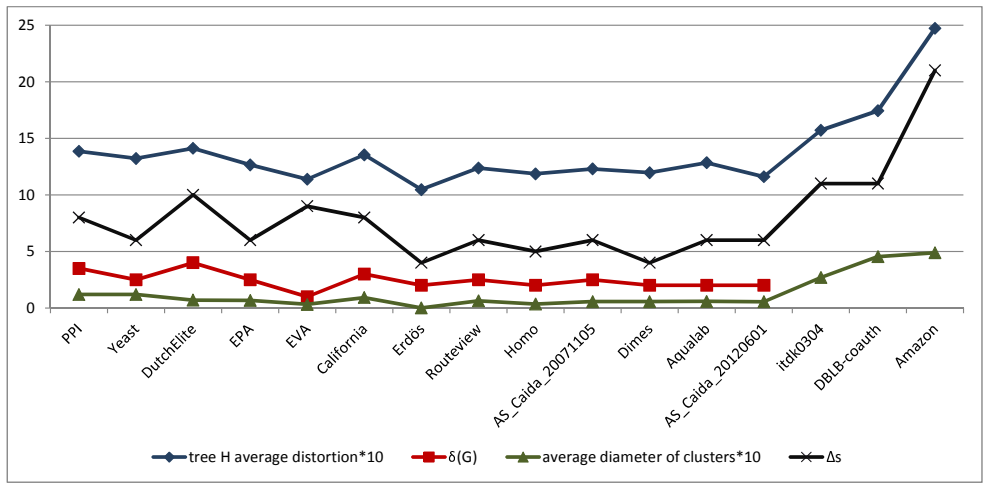}
               \caption{Four tree-likeness measurements scaled.}\label{fig:tree-likeness-1charts}
\end{minipage}
\end{center}
\end{figure}

\begin{figure}
\vspace*{-1.4cm}
        \centering
        \begin{subfigure}[b]{0.450\textwidth}
                \includegraphics[width=\textwidth]{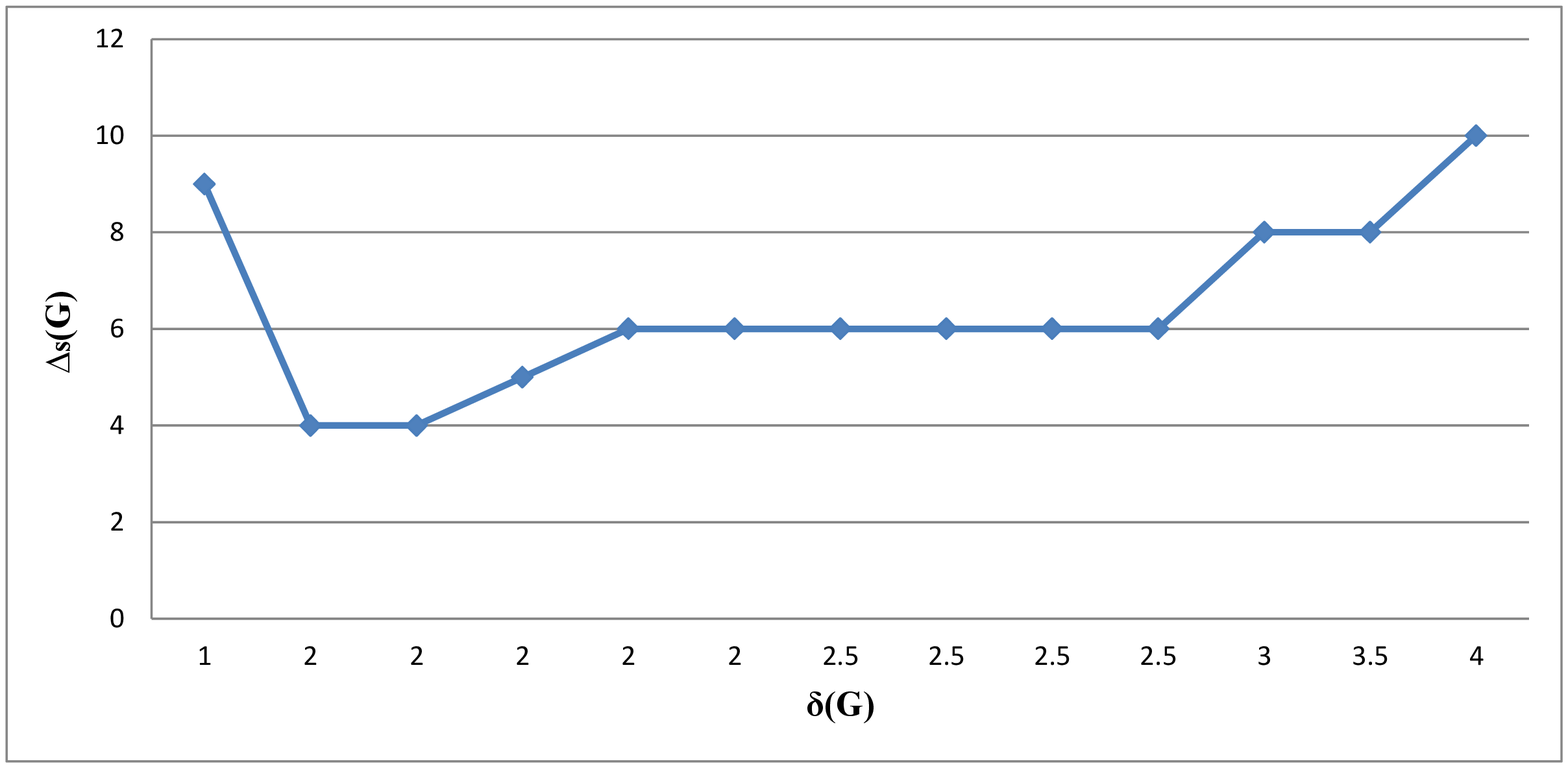}
                \caption{hyperbolicity $\delta(G)$ vs. cluster-diameter $\Delta_s(G).$}
                \label{fig:delta-clusDiam}
        \end{subfigure}
        \begin{subfigure}[b]{0.450\textwidth}
                \includegraphics[width=\textwidth]{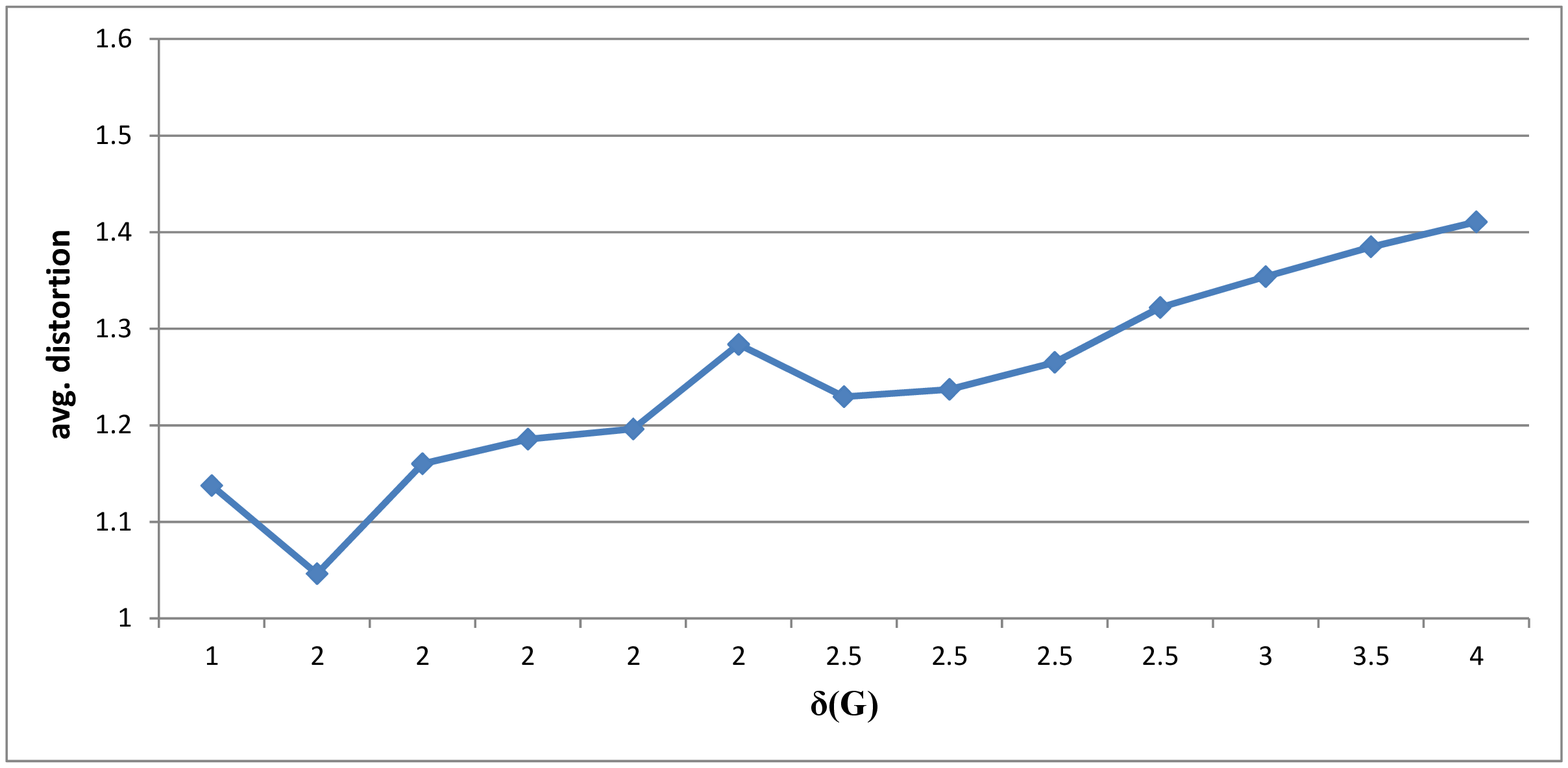}
                \caption{hyperbolicity $\delta(G)$ vs. avg. distortion of $H.$}
                \label{fig:deltaHdist}
        \end{subfigure}
         \begin{subfigure}[b]{0.450\textwidth} 
                \includegraphics[width=\textwidth]{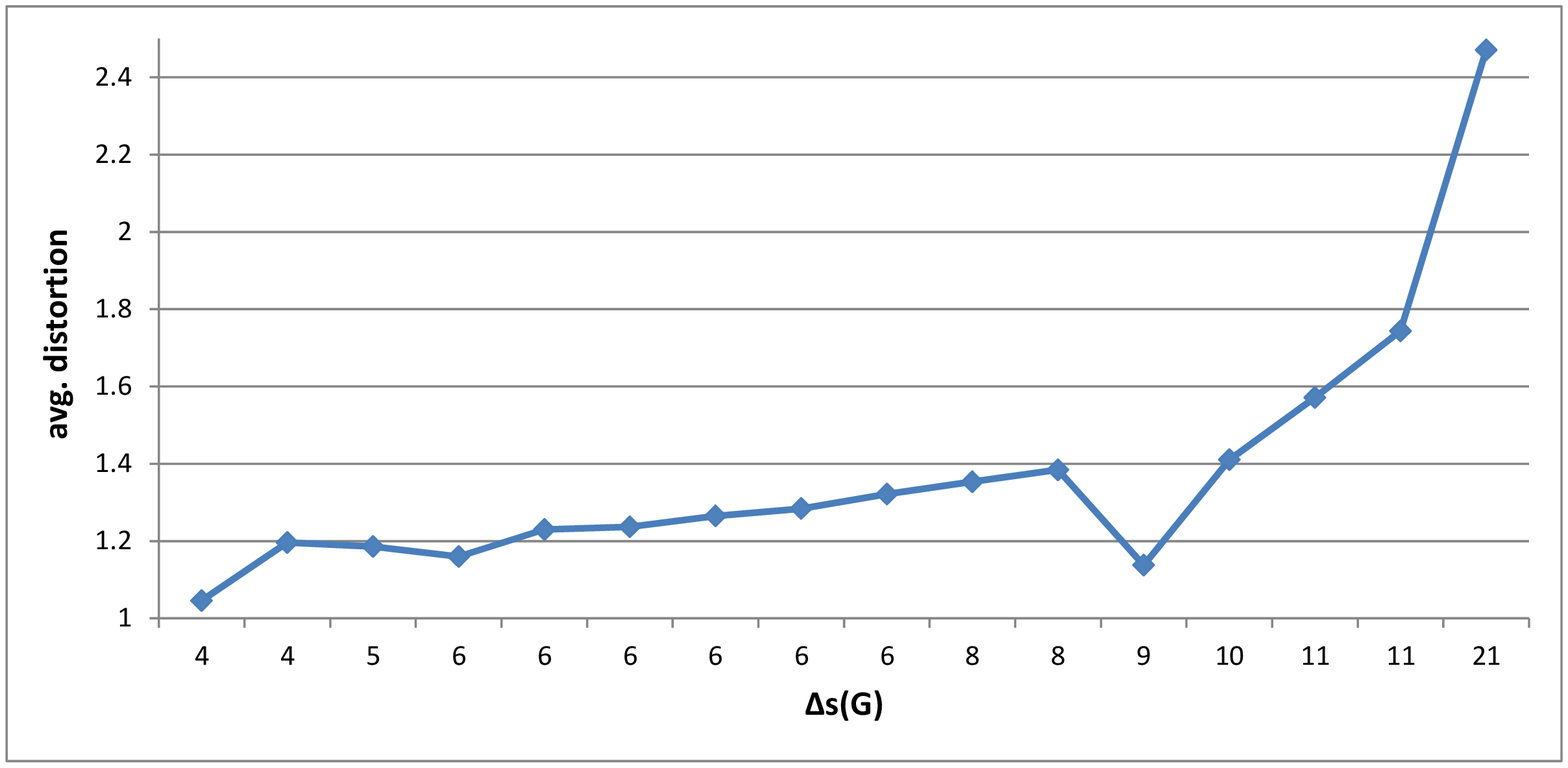}
                \caption{cluster-diameter $\Delta_s(G)$ vs. avg. distortion of $H.$}
                \label{fig:clusDiam-Hdist}
        \end{subfigure}
        \begin{subfigure}[b]{0.450\textwidth}
                \includegraphics[width=\textwidth]{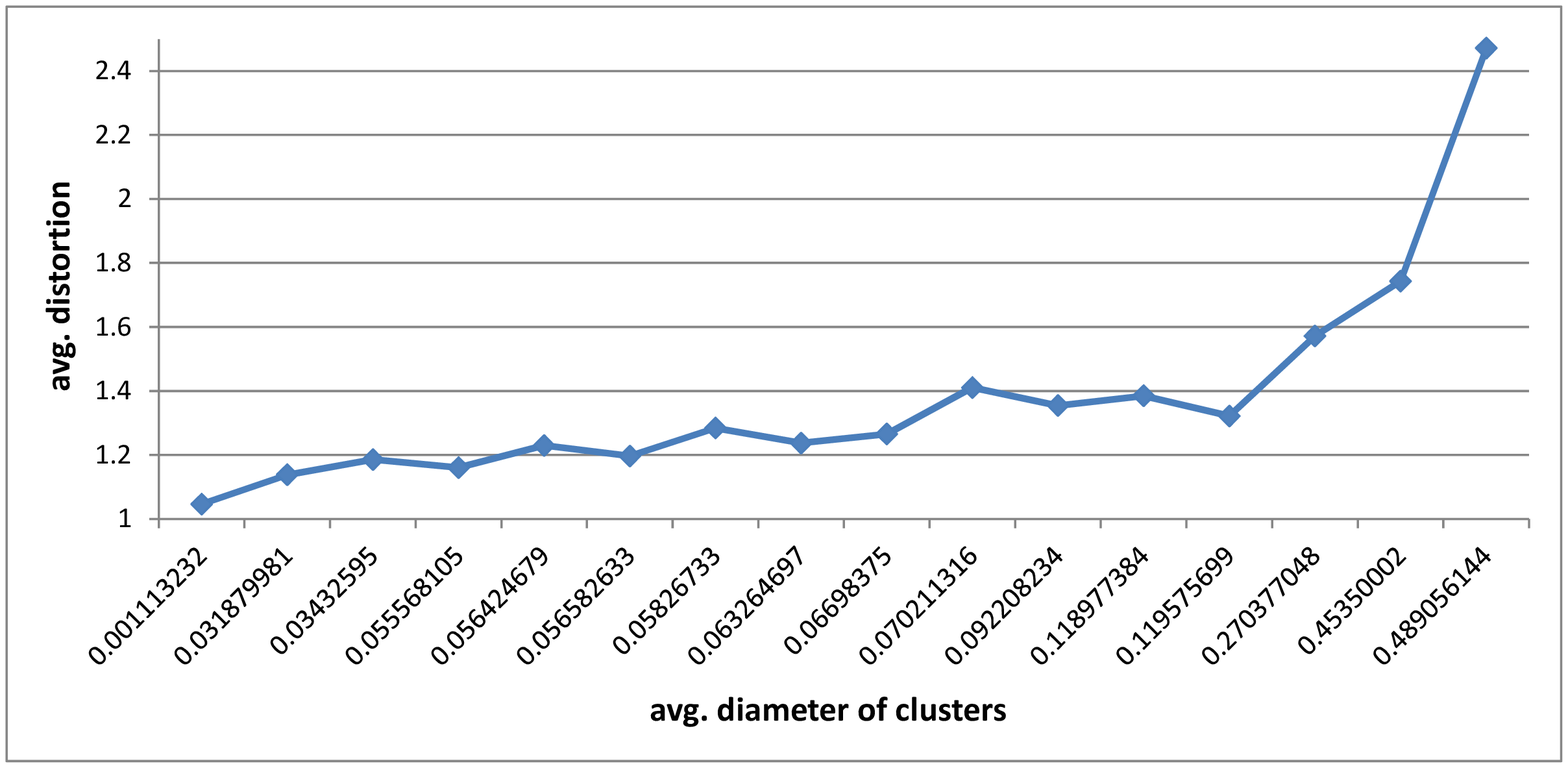}
                \caption{avg. diameter of clusters vs. avg. distortion of $H.$} 
                \label{fig:avgClus-diamHdist}
        \end{subfigure}
        \caption{Tree-likeness measurements: pairwise comparison.}\label{fig:tree-likeness-charts}
\end{figure}

\clearpage


%

\end{document}